\documentclass[12pt, reqno]{amsart}
\usepackage{mathrsfs}
\usepackage{amsmath}
\usepackage[height=22cm, width=16.5cm, hmarginratio={1:1}]{geometry}
\usepackage{diagbox}
\usepackage{tikz}
\usetikzlibrary{positioning, shapes.geometric}
\usepackage{newclude}
\usepackage{appendix}
\usepackage{extarrows}

\def\sspan{\mathrm{span}}
\usepackage[hidelinks,hyperindex,breaklinks]{hyperref}
\usepackage{tikz-cd}
\usetikzlibrary{positioning, arrows,decorations.pathmorphing,shapes}
\usetikzlibrary{shapes.geometric}
\usepackage{array}

\author{Shuai Guo}
\address{School of Mathematics Science, Peking University,
      No 5. Yiheyuan Road, Beijing, 100871, China}
\email{guoshuai@math.pku.edu.cn}

\author{Qingsheng Zhang}
\address{School of Sciences, Great Bay University,
     Great Bay Institute for Advanced Study,
     Dongguan, 523000, China}
\email{zhangqingsheng@gbu.edu.cn}

\newcommand{\<}{\langle}
\renewcommand{\>}{\rangle}

\newcolumntype{C}[1]{>{\centering\arraybackslash$}m{#1}<{$}}
\newlength{\mycolwd}                                         
\newlength{\mycolwdm}                                         
\newlength{\mycolwda}                                         
\newlength{\mycolwdb}                                         
\newlength{\mycolwdc}                                         
\newlength{\mycolwdd}                                         
\newlength{\mycolwddd}                                         
\newlength{\mycolwddc}                                         
\newcommand{\givz}{\mathfrak{u}}
\newcommand{\givw}{\mathfrak{v}}
\newcommand{\Frob}{H}
\newcommand{\cyA}{\mathfrak A}
\newcommand{\cyB}{\mathfrak B}
\newcommand{\ttau}{t}

\newcommand{\bm}{\frak{m}}

\newcommand{\vac}{\mathbf{v}}
\newcommand{\dvac}{\nu}
\newcommand{\cmn}{{\substack{\mu\\ \nu}}}
\newcommand{\sv}{\varphi}

\newcommand{\elltau}{\tau}

\newcommand{\id}{ \hspace{-2pt}\text{I}}
\usepackage {color,graphicx,psfrag, verbatim,amssymb,amscd,enumerate,subfigure}

\def\End{\mathrm{End}}

\def\beq{\begin{equation}}
\def\eeq{\end{equation}}
\newcommand{\be}{\begin{equation*}}
\newcommand{\ee}{\end{equation*}}

\DeclareMathOperator{\Cont}{Cont}

\DeclareMathOperator{\Res}{Res}

\DeclareMathOperator{\diag}{diag}

\def\Mgn{\overline{\mathcal{M}}_{g,n}}

\newtheorem{dummy}{}[section]
\newtheorem{lemma}[dummy]{Lemma}
\newtheorem{proposition}[dummy]{Proposition}
\newtheorem{theorem}[dummy]{Theorem}

\theoremstyle{definition}

\newtheorem{definition}[dummy]{Definition}

\newtheorem{remark}[dummy]{Remark}

\usepackage{graphicx}

\newcommand{\cL}{\mathcal L}

\newcommand{\tr}{\mathrm{tr}}

\newcommand{\pd}{\partial}

\newcommand{\cA}{\mathcal{A}}

\newcommand{\cC}{\mathcal{C}}
\newcommand{\cD}{\mathcal{D}}
\newcommand{\E}{\mathcal{E}}

\newcommand{\Mbar}{\overline{\mathcal M}}

\newtheorem{manualtheoreminner}{Theorem}
\newenvironment{manualtheorem}[1]{%
	\renewcommand\themanualtheoreminner{#1}%
	\manualtheoreminner
}{\endmanualtheoreminner}

\begin{document}
	\title{Virasoro constraints for topological recursion}
	
	\maketitle

\begin{abstract}
This is the second paper in a series  on {\it Virasoro constraints for Cohomological  Field Theory}. 
We derive the ancestor Virasoro constraints for the topological recursion (TR) for an arbitrary spectral curve and establish the descendent Virasoro constraints for spectral curves satisfying certain conditions.
For higher-genus curves, we further establish the corresponding ancestor and descendent Virasoro constraints for the associated non-perturbative generating series.
We present several examples that illustrate the comparison between the descendent Virasoro constraints for TR descendent invariants and the original Virasoro constraints for geometric descendent invariants.
\end{abstract}
	
\setcounter{section}{-1}
\setcounter{tocdepth}{1}
\tableofcontents

\section{Introduction}
This is the second paper in a three-part series on Virasoro constraints in cohomological field theory (CohFT). In the first paper~\cite{GZ25}, we introduced a formal descendent generating series for $S$- and $\dvac$-calibrated CohFTs with vacuum, and formulated both the ancestor and descendent Virasoro conjectures for the homogeneous case. In particular, we established these two conjectures for semisimple homogeneous CohFTs with vacuum.

The topological recursion was introduced by Eynard--Orantin~\cite{EO07} in the context of matrix models (see, e.g., ~\cite{CEO06,EO07b}).
Subsequent studies have established that topological recursion  is deeply linked to CohFTs and enumerative geometry problems: the former corresponds to the ancestor potential, and the latter to the descendent potential.
The identification of a semisimple CohFT with (local) topological recursion was established by Dunin-Barkowski, Orantin, Shadrin, and Spitz~\cite{DOSS14}; see also the works of Eynard~\cite{Eyn14} and Milanov~\cite{Mil14}.
On the descendent side, the connection between topological recursion and enumerative geometry has been investigated via numerous examples, including Hurwitz numbers~\cite{BM07,BEMS11}, Gromov--Witten theory of $\mathbb P^1$~\cite{NS11,DOSS14}, negative $r$-spin theory~\cite{CGG22,GJZ23} and toric Calabi-Yau 3-folds~\cite{BKMP09, EO15,FLZ20}, among others.
In~\cite[Theorem II]{GJZ23}, 
the authors, together with C. Ji, established a general relationship between the formal descendent theory of $S$- and $\dvac$-calibrated CohFTs and topological recursion via Laplace transformations for spectral curves satisfying certain conditions.
In the language of generating series, this establishes an explicit relationship between the geometric total descendent potential $\cD({\bf t};\hbar)$ and the TR descendent potential $Z({\bf p};\hbar)$ through a special choice of coordinates ${\bf t=t(p)}$,  meaning that the TR descendent potential is typically a generating function of a  sub-theory of certain enumerative geometry. This correspondence serves as the first motivation for our investigation into the Virasoro constraints for topological recursion. 
   
Moreover, the generating series $Z({\bf p};\hbar)$ arising from the topological recursion for genus-zero curves, along with its non-perturbative counterpart \(Z^{\rm NP}_{\mu,\nu}({\bf p};\hbar;w)\) for higher-genus curves, is expected to be the tau-function of certain reductions of the multi-component KP hierarchy~\cite{ABDKS25, ADKMV06, BE15, GJZ23}. Topological recursion plays a crucial role in our generalization of the Witten conjecture \cite{GJZ23} (see \cite{Wit91, Kon92} for the original conjecture).
This connection serves as the secondary motivation for our investigation into the Virasoro constraints for topological recursion.

\subsection{Topological recursion}
\label{sec:def-TR}

Let $\Sigma$ be a Riemann surface of genus $\frak g$ with marked points (boundaries) $\{b_i\}_{i=1}^{\bm}$ and a fixed symplectic basis $\{\cyA_i,\cyB_i\}_{i=1}^{\frak g}$ for $H_{1}(\Sigma,\mathbb C)$.
$x=x(z)$ is a (possibly multi-valued) function on $\Sigma\setminus\{b_i\}$ whose differential is meromorphic, with poles exactly at the boundaries and only finitely many  simple zeros $\{z^\beta\}_{\beta=1}^{N}$, called the critical points.
$y=y(z)$ is a holomorphic function defined near each critical point $z^\beta$ satisfying $dy(z^\beta)\ne 0$.
For any point $z\in\Sigma$ around the critical point $z^\beta$, there exists a unique local involution $\bar{z}\in \Sigma$ (where $\bar{z}$ depends locally on $z$) such that $x(\bar{z})=x(z)$. 
For more details, we refer the reader to~\cite{EO07}.
We now recall the definition of the sequence of multi-differentials $\{\omega_{g,n}\}_{g,n\geq 0}$ associated to a spectral curve $\cC=(\Sigma,x,y)$.
\begin{definition}\label{def:omegagn}
	Given spectral curve $\cC=(\Sigma,x,y)$, a family of multi-differentials $\{\omega_{g,n}\}_{g,n\geq 0}$ is defined as follows. For $2g-2+n\leq 0$, we set $\omega_{0,0}(z):=0$, $\omega_{1,0}(z):=0$, and
	$$
	\omega_{0,1}(z):=y(z)dx(z), \quad \omega_{0,2}(z_1, z_2):=B(z_1, z_2),
	$$
where $B(z_1, z_2)$ is the Bergman kernel of the Riemann surface $\Sigma$.
	Let $K_\beta$ be recursion kernel  around the critical point $z^\beta$, defined by
	$$
	K_\beta(z_0,z)=\frac{\int_{z'=\bar z}^{z}B(z_0,z')}{2(y(z)-y(\bar z))dx(z)}.
	$$
	For $2g-2+n+1>0$, the multi-differentials $\omega_{g,n+1}$ are defined recursively as follows:
	\begin{align}
		\omega_{g,n+1}(z_0,z_{[n]})
		:=&\, \sum\limits_{\beta}\mathop{\Res}_{z=z^\beta}K_\beta(z_0,z)\bigg(\omega_{g-1,n+2}(z,\bar z,z_{[n]}) \nonumber\\
		&\, +\sum^{\prime}_{\substack{g_1+g_2=g\\ I\bigsqcup J=[n]}}\omega_{g_1,|I|+1}(z,z_{I})\omega_{g_2,|J|+1}(\bar z,z_{J})\bigg).
\label{eqn:def-TR}
	\end{align}
	Here $[n]=\{1,\cdots,n\}$ and for any subset $I\subset [n]$, $z_{I}:=\{z_i\}_{i\in I}$.
	The symbol $\sum\limits^{\prime}$ means we exclude $\omega_{0,1}$ in the summation.
	For $g\geq 2$,
	$$
	\omega_{g,0}:=\frac{1}{2-2g}\sum_{\beta}\mathop{\Res}_{z=z^\beta}\omega_{g,1}(z)\int_{z'=z^\beta}^{z}\omega_{0,1}(z').
	$$
\end{definition}
\begin{remark}
	In~\cite{EO07}, the $0$-forms $\omega_{0,0}$ and $\omega_{1,0}$ are also defined (cf.~(4.17) and (4.18))  and are not necessarily zero. Since they do not affect the validity of the Virasoro constraints, we assume them to vanish for simplicity.
\end{remark}

When expanded under a suitable basis,  the coefficients of   $\omega_{g,n}$  can be  viewed as the ``invariants" of the spectral curve data {with respect to this basis}.
Inspired by the enumerative geometry where the ancestor invariants and descendent invariants are defined, we consider the following two types of invariants for a topological recursion.
\begin{enumerate}[(i)]
\item {\bf Ancestors invariants:}
Using the local Airy coordinates near the critical points, we can  define a set of meromorphic $1$-forms $\{d\zeta_k^{\bar\beta}(z)\}$ as a basis.
For $2g-2+n>0$, the expansion of $\omega_{g,n}$ under this basis defines the ancestor invariants $\<\bar e_{\beta_1}\bar\psi^{k_1},\cdots,\bar e_{\beta_n}\bar\psi^{k_n}\>_{g,n}$. 
In fact, these invariants coincide with the ancestor invariants of the CohFT associated with the spectral curve of the topological recursion.
(See \S \ref{subsec:TR-CohFT-ancestors}.)
\item{\bf Descendents invariants:}
Pick $\Lambda=(\lambda_1,\cdots,\lambda_{\bm})$, such that $\lambda_i^{-1}$ is a local coordinate near the boundary point $b_i$ satisfying $\lambda_i(b_i)={\infty}$, $i=1,\cdots,\bm$.
By using basis $\{d\lambda^{-k}_{i}\}$, for $g,n\geq 0$, the expansion of $\omega_{g,n}$ near the boundary points defines the descendent invariants
$\<\alpha^{i_1}_{k_1},\cdots,\alpha^{i_n}_{k_n}\>^{\Lambda}_{g,n}.$
(See \S \ref{subsec:TR-descendents}.)
\end{enumerate}
\begin{remark}
	Throughout this paper, descendents, or descendent invariants, refer to the TR descendents in \cite{GJZ23},  unless explicitly stated otherwise. Geometric descendents will be explicitly specified when discussed.
\end{remark}

The generating series of the ancestors $\<-\>_{g,n}$ and descendents $\<-\>^{\Lambda}_{g,n}$ are denoted by $\cA({\bf s};\hbar)$ and $Z^{\Lambda}({\bf p};\hbar)$, respectively. Precisely,
$$
\cA({\bf s};\hbar):=\exp\biggl(
\sum_{2g-2+n> 0}\hbar^{2g-2}\sum_{1\leq \beta_1,\cdots,\beta_n\leq N \atop k_1,\cdots,k_n\geq 0}
\<\bar e_{\beta_1}\bar\psi^{k_1},\cdots,\bar e_{\beta_n}\bar\psi^{k_n}\>_{g,n}
\frac{s_{k_1}^{\bar\beta_1}\cdots s_{k_n}^{\bar\beta_n}}{n!}
\biggr),
$$
and
$$
Z^{\Lambda}({\bf p};\hbar)
:=\exp\bigg(\sum_{g,n\geq 0}\hbar^{2g-2}\sum_{1\leq i_1,\cdots,i_n\leq \bm \atop k_1,\cdots,k_n\geq 1}
\<\alpha^{i_1}_{k_1},\cdots,\alpha^{i_n}_{k_n}\>_{g,n}^{\Lambda}
\frac{p^{i_1}_{k_1}\cdots p^{i_n}_{k_n}}{n!\cdot k_1\cdots k_n}\bigg).
$$

\subsection{Ancestor Virasoro constraints for topological recursion}
For the ancestor generating series $\cA({\bf s};\hbar)$, the ancestor Virasoro constraints 
$$
L_m\cA({\bf s};\hbar)=0,\qquad m\geq -1,
$$ 
can be derived via the identification of  (local) topological recursion with semisimple CohFTs, as established in~\cite{DOSS14} (see also~\cite{Eyn14, Mil14}), along with arguments similar to those in~\cite{GZ25} and~\cite{Ale24}.
Explicit formulae for the operators $L_m$ are  provided in \S\ref{sec:TR-ancestor-virasoro}.

In the present paper, we provide a direct derivation of the ancestor Virasoro constraints from the topological recursion.
We have
\begin{lemma}\label{lem:TR-res-intro}
For $2g-2+n+1>0$ and $m\geq -1$, we have
\begin{align}
 &\, \sum_{\beta}\mathop{\Res}_{z_0=z^\beta} x(z_0)^{m+1}\cdot y(z_0)\cdot \omega_{g,n+1}(z_0,z_{[n]})\nonumber \\
=&\,  \frac{1}{2}\sum_{\beta}\mathop{\Res}_{z=z^\beta} \frac{x(z)^{m+1}}{dx(z)}\cdot \bigg(\omega_{g-1,n+2}(z,\bar z,z_{[n]})
+\sum^{\prime}_{\substack{g_1+g_2=g\\ I\sqcup J=[n]}} \omega_{g_1,|I|+1}(z,z_{I})\omega_{g_2,|J|+1}(\bar z,z_J)\bigg).
\label{eqn:TR-virasoro}
\end{align} 
\end{lemma}

\begin{proposition}\label{prop:TR-virasoro-ancestor}
For each $m\geq -1$, the ancestor Virasoro constraint $L_m\cA({\bf s};\hbar)=0$ is equivalent to equation~\eqref{eqn:TR-virasoro}, and hence holds for the topological recursion.
\end{proposition} 
By the equivalence  in Proposition \ref{prop:TR-virasoro-ancestor}, we call equation~\eqref{eqn:TR-virasoro} the {\it ancestor Virasoro constraints for TR-differentials}. 
\begin{remark}
In~\cite{Mil14}, Milanov proved that, in the context of singularity theory, the local topological recursion is equivalent to $N$ copies of the Virasoro constraints for the ancestor generating series.
Moreover, his method can be naturally generalized to general settings.
Proposition~\ref{prop:TR-virasoro-ancestor} can be regarded as a global version of Milanov’s result.
\end{remark}

\subsection{Descendent Virasoro constraints for topological recursion}
In the geometric theory, one can derive the descendent Virasoro constraints from the ancestor Virasoro constraints by using the $S$-matrix~\cite{GZ25}.
But this approach does not apply to the topological recursion since, in general, there is no such a invertable matrix that relates the TR descendents and TR ancestors.
The bridge connecting these two types of invariants is the family of globally defined multi-differentials whose local expansions encode the corresponding invariants.
So we first derive the {\it descendent Virasoro constraints for TR-differentials}
from equation~\eqref{eqn:TR-virasoro}.
This can be achieved by using the residue formula, under the following conditions:
First, $x$, $y$ are meromorphic on the Riemann surface $\Sigma$;
Second, $x^{m+1}y$ has no poles other than those of $x$;
Third, the locally defined variable $\bar z$ is properly treated.
We have the following result.
\begin{proposition}\label{prop:TR-crit-boundary-intro}
If $x$, $y$ are meromorphic and satisfy that $x^{m+1}y$ has only poles at the boundary points, then for $2g-2+n+1>0$, we have
\begin{align}
 &\, -\sum_{i}\mathop{\Res}_{z_0=b_i} x(z_0)^{m+1}\cdot y(z_0)\cdot \omega_{g,n+1}(z_0,z_{[n]})
\nonumber \\
=&\,  \frac{1}{2}\sum_{i}\mathop{\Res}_{z=b_i} \frac{x(z)^{m+1}}{dx(z)}\cdot \bigg(\tilde\omega_{g-1,n+2}(z, z,z_{[n]})
+\sum^{\prime}_{\substack{g_1+g_2=g\\ I\sqcup J=[n]}} \omega_{g_1,|I|+1}(z,z_{I})\omega_{g_2,|J|+1}(z,z_J)\bigg).
\end{align}
where $\tilde\omega_{g,n}(z_{[n]})=\omega_{g,n}(z_{[n]})-\delta_{(g,n),(0,2)}\cdot\frac{dx(z_1)dx(z_2)}{(x(z_1)-x(z_2))^2}$.
\end{proposition}

For the generating series of TR descendents, to make the formula concise, we consider a special choice of local coordinate system $\Lambda=(\lambda_1,\cdots,\lambda_\bm)$ as follows.
Near each boundary point $b_i$, there exist a unique positive integer $r_i$ and a unique choice (up to a factor $\xi_i$ such that $\xi_i^{r_i}=1$) of the local coordinate $\lambda_i$ such that $x(z)=\lambda_i^{r_i}$.
For this choice of the local coordinate system $\Lambda$, we denote the corresponding generating series $Z^{\Lambda}({\bf p};\hbar)$ by $Z({\bf p};\hbar)$.

\begin{manualtheorem}{1}\label{thm:TR-virasoro-descendent-intro}
Suppose $x$, $y$ are meromorphic functions on $\Sigma$.
For $m\geq -1$, if $x^{m+1}y$ has only poles at the boundary points, then the descendent generating series $Z({\bf p};\hbar)$ for the topological recursion satisfies
$$
\cL_mZ({\bf p};\hbar)=0.
$$
Here, let $c_m=-\sum_{z^0: \, x(z^0)=0}\Res_{z=z^0} x(z)^{m+1}y(z)\omega_{0,1}(z)$ and $\tilde p_k^i=p_k^i-v_k^i$, where the non-vanished constants $v_k^i$ come from $y(z)dx(z)|_{b_i}=\sum_{k\geq 0} v_k^i \lambda_i^{k-1} d\lambda_i +\sum_{k\geq 1}\<\alpha^i_k\>_{0,1}\frac{d\lambda_i^{-k}}{k}$,
then the operator $\cL_m$ is given by
\begin{align*}
\cL_m=&\, \frac{c_m}{2\hbar^2}+
\frac{\delta_{m,-1}}{2\hbar^2}\sum_{i}\frac{1}{r_i}\sum_{a=0}^{r_i}\tilde p^i_{a}\tilde p^i_{r_i-a}
+\frac{\delta_{m,0}}{2\hbar^2}\sum_{i}\frac{(v_0^i)^2}{r_i} 
+\delta_{m,0}\sum_{i}\frac{r_i^2-1}{24r_i} \\
&\, +\sum_{i}\frac{1}{r_i}\sum _{k\geq 0} \tilde p^i_k (r_im+k)\frac{\pd}{\pd p^i_{r_im+k}} 
+\frac{\hbar^2}{2}\sum_{i}\frac{1}{r_i}\sum_{k+l=r_im}kl\frac{\pd^2}{\pd p^i_k\pd p^i_l},
\end{align*}
where the summation $\sum_i$ is over the indices labeling the boundary points.
\end{manualtheorem}
\begin{remark}
We will see in \S\ref{sec:TR-descendent-virasoro} that Theorem~\ref{thm:TR-virasoro-descendent-intro} is derived from Proposition~\ref{prop:TR-crit-boundary-intro}.
This partially explains why the TR descendent Virasoro constraints should be considered for invariants defined near the boundary points.
\end{remark}

For $\frak{g}(\Sigma)>0$, one can also consider the non-perturbative invariants from the topological recursion and define the non-perturbative generating series $Z^{\Lambda,\rm NP}_{\mu,\nu}({\bf p};\hbar;w)$ and $Z^{\rm NP}_{\mu,\nu}({\bf p};\hbar;w)$ for above special choice of local coordinate system $\Lambda$, see \S \ref{sec:NP-TR-virasoro}.
In this paper, we prove that the non-perturbative generating series $Z^{\rm NP}_{\mu,\nu}({\bf p};\hbar;w)$ satisfies the same Virasoro constraints as the ones for $Z({\bf p};\hbar)$.
\begin{manualtheorem}{2}\label{thm:NP-TR-virasoro-descendent-intro}
Suppose $x$, $y$ are meromorphic functions on $\Sigma$.
For $m\geq -1$, if $x^{m+1}y$ has only poles at the boundary points, then the non-perturbative descendent generating series $Z^{\rm NP}_{\mu,\nu}({\bf p};\hbar;w)$ for the topological recursion satisfies
$$
\cL_mZ^{\rm NP}_{\mu,\nu}({\bf p};\hbar;w)=0,
$$
where the operator $\cL_m$ is the same one as given in Theorem~\ref{thm:TR-virasoro-descendent-intro}.
\end{manualtheorem}

\subsection{Outline of the paper} 
This paper is organized as follows.  
In \S \ref{sec:CohFT-TR}, we review the CohFT associated with the spectral curve.
In \S \ref{sec:TR-ancestor-virasoro} and \S \ref{sec:TR-descendent-virasoro},
we derive the ancestor Virasoro constraints and descendent Virasoro constraints for the topological recursion, respectively.
In \S \ref{sec:NP-TR-virasoro}, we prove the Virasoro constraints for the non-perturbative topological recursion.
In \S\ref{sec:example}, we present several examples to compare the Virasoro constraints for TR descendents with those for geometric descendents.

\textbf{Acknowledgements.}
This work is partially supported by the National Key Research and Development Program of China 2023YFA1009802 and the National Natural Science Foundation of China 12225101.

\section{CohFT associated with the spectral curve}
\label{sec:CohFT-TR}
In this section, we first review the notation for CohFTs with vacuum, then recall the CohFT associated with a spectral curve, and conclude by introducing the homogeneity condition for the spectral curve.

\subsection{CohFT with vacuum}\label{subsec:CohFT}
We begin by recalling the definition of a CohFT~\cite{KM94}.
Let  $\Frob$ be a finite dimensional vector space over a field $\mathbb F$ (we take $\mathbb F=\mathbb C$) with a non-degenerate symmetric 2-form $\eta$. 
Let $\{\phi_i\}_{i=1}^{N}$ be a basis of $\Frob$, where $N$ is the dimension of $\Frob$,  we denote by $\{\phi^i\}_{i=1}^{N}$ the dual basis of $\{\phi_i\}_{i=1}^{N}$ with respect to $\eta$.
A CohFT $\Omega=\{\Omega_{g,n}\}_{2g-2+n>0}$ on $(\Frob, \eta)$ is a set of maps to the cohomological classes (taking value in an $\mathbb F$-algebra $\mathbb A$) on the moduli space $\Mbar_{g,n}$ of stable curves:
$$
\Omega_{g, n} \colon \Frob ^{\otimes n}  \rightarrow   H^*(\Mbar_{g,n},\mathbb{A}),
$$
that satisfies 1). {\it $S_n$-invariance axiom}: $\Omega_{g,n}(v_1,\cdots,v_n)$ is invariant under permutation between $v_i\in \Frob$, $i=1,\cdots,n$, corresponding the marked points of $\Mbar_{g,n}$,
and 2). {\it gluing axiom}: the pull-backs $q^*\Omega_{g,n}$ and $r^*\Omega_{g,n}$ of the gluing maps
\begin{equation*}
q\colon \Mbar_{g-1, n+2}\to \Mbar_{g,n} \quad {\text{ and }}\quad
r\colon \Mbar_{g_1, n_1+1}\times \Mbar_{g_2, n_2+1} \to \Mbar_{g_1+g_2, n_1+n_2}
\end{equation*}
 are equal to the contraction
of $\Omega_{g-1, n+2}$ and $\Omega_{g_1, n_1+1} \otimes \Omega_{g_2, n_2+1}$ by the bivector $\sum_i \phi_i\otimes \phi^i$.
If there is a distinguished element $\vac(\givz)=\sum_{m\geq 0}\vac_m \givz^m\in \Frob[[\givz]]$, called the vacuum vector, such that
$$
 \pi^*\Omega_{g,n}(v_1, \cdots , v_n)=\Omega_{g,n+1}\big(v_1, \cdots , v_n, \vac(\psi_{n+1}) \big), \qquad 2g-2+n>0,
$$
where $\pi\colon\Mbar_{g,n+1}\to \Mbar_{g,n}$ is the forgetful map defined by forgetting the last marked point,
then we call the CohFT satisfies the {\it vacuum axiom}, and such a CohFT is called the CohFT with vacuum.
We refer the reader to~\cite{PPZ15,CGL18} for more discussions on CohFTs.

The ancestor correlator $\<-\>_{g,n}^{\Omega}$ for the CohFT $\Omega$ is defined as follows: for $2g-2+n>0$, 
\beq\label{def:bracket-anc}\textstyle
\<v_1\bar\psi^{k_1},\cdots,v_n\bar\psi^{k_n}\>^{\Omega}_{g,n}
:=\int_{\Mbar_{g,n}}\Omega_{g,n}(v_1, \cdots , v_n)\psi_1^{k_1}\cdots\psi_n^{k_n},
\eeq
and for $2g-2+n\leq 0$, $\<-\>_{g,n}^{\Omega}:=0$.
We usually drop the superscript $\Omega$ in $\<-\>_{g,n}^{\Omega}$ if no confusion arises.
The total ancestor potential $\cA({\bf s};\hbar)$ is defined by
\beq\label{def:FandA-CohFT}
\cA({\bf s};\hbar):=\exp\bigg(
\sum_{g\geq 0}\hbar^{2g-2}\sum_{n\geq 0}\frac{1}{n!}\<{\bf s}(\bar\psi_1),\cdots,{\bf s}(\bar\psi_n)\>_{g,n}
\bigg),
\eeq
where ${\bf s}(\bar\psi)=\sum_{k,i} s^i_k \, \phi_i\bar\psi^k $ with formal variables $\{s_k^i\}$.

\subsection{Semisimple CohFT and its classification}
\label{subsec:CohFT-reconstruction}
A CohFT $\Omega$ induces a commutative and associative product structure, called the {\it quantum product},  $*$ on $\Frob$, which is defined by
$$
\eta(v_1* v_2, v_3)=\Omega_{0,3}(v_1, v_2, v_3).
$$
One can prove that ${\bf 1}:=\vac_0$ is the unity of the quantum product: ${\bf 1}*v=v$, $v\in \Frob$.
We call the vacuum vector is flat if $\vac(\givz)={\bf 1}$, and we call such element ${\bf 1}$ the flat unit.

A CohFT is called {\it semi-simple} if the algebra $(\Frob,\eta,*)$ can be decomposed into $N$ one dimensional algebras $(\mathbb E\{e_\beta\},\eta,*)$
where $e_\beta*e_\beta=e_\beta$ and $\eta(e_\beta,e_\beta)=\Delta_{\beta}^{-1}$ for some non-vanished elements $\Delta_{\beta}$. 
Here $\mathbb E$ is an extended field of $\mathbb F$.
We call $\{\phi_i\}$ and $\{e_\beta\}$ the flat basis and the canonical basis (of $H$ or of the CohFT), respectively.
We define $\bar e_\beta=\Delta_\beta^{\frac{1}{2}}e_\beta$ and call $\{\bar e_\beta\}$ the normalized canonical basis.
The relation of these bases are given by
$\phi_i= \Psi^{\bar\beta}_i \bar e_{\beta}=\widetilde\Psi^{\bar\beta}_i e_{\beta}$. 
Clearly, $\widetilde\Psi^{\beta}_i=\Delta_\beta^{\frac{1}{2}}\Psi^{\bar\beta}_i$.
We define formal variables $\{s^{\bar \beta}_k\}_{k,\beta}$ by $s_k^i\phi_i=s_k^{\bar\beta}\bar e_\beta$.

Based on Givental's works~\cite{Giv01a,Giv01b}, Teleman~\cite{Tel12} proved that a semi-simple CohFT $\Omega$ with vacuum $\vac(\givz)$ can be uniquely reconstructed from $N$ trivial CohFTs $\Omega^{\rm KW_{\beta}}$ (for which case, the state space is $\mathbb F\{\bar e_\beta\}$ with the symmetric $2$-form given by $\eta(\bar e_\beta,\bar e_\beta)=1$ and $\Omega^{\rm KW_{\beta}}_{g,n}(\bar e_\beta,\cdots,\bar e_\beta)=1$) by the following formula:
\beq\label{eqn:Giv-Tel}\textstyle
\Omega=R\cdot T\cdot (\oplus_{\beta=1}^{N}\Omega^{\rm KW_{\beta}}).
\eeq
The formula is explained as follows.
Firstly, $R(\givz)=\sum_{k\geq 0}R_k\givz^k\in \mathrm{End}(\Frob)[[\givz]]$ is an formal power series that
satisfies the symplectic condition
$
R^*(\givz)R(-\givz)={\rm Id},
$
where $R^*(\givz)$ is the adjoint of $R(\givz)$ with respect to $\eta$. 
We define
$V(\givz, \givw)
=\frac{{\rm Id}-R^*(-\givz)R(-\givw)}{\givz+\givw}.$
The $R$-action $R\cdot $ on a CohFT $\Omega'$ is defined as follows.
Let $\mathcal{G}_{g,n}$ be the set of stable graphs of genus $g$ with $n$ legs. For $\Gamma\in \mathcal{G}_{g,n}$,  define
$
{\rm Cont}_\Gamma\colon \Frob ^{\otimes n}  \rightarrow   H^*(\Mbar_{g,n},\mathbb{A})
$
by the following construction:
\begin{itemize}
  \item[1.] placing $\Omega'_{g(v), n(v)}$ at each vertex $v\in \Gamma$,
  \item[2.] placing $R^*(-\psi_i)\cdot$ at each leg $l_i \in \Gamma$ labeled by $i=1,\ldots, n$,
  \item[3.] placing $V(\psi_{v'}, \psi_{v''})\phi_i\otimes \phi^i$ at each edge $e\in \Gamma$ connection vertexes $v'$ and $v''$,
\end{itemize}
Then $R\cdot\Omega'$ is defined to be the  CohFT on the state space $(\Frob,\eta)$ with the maps
$$\textstyle
(R\cdot\Omega')_{g,n}=\sum_{\Gamma\in \mathcal{G}_{g,n}}\frac{1}{\left|{\rm Aut}(\Gamma)\right|}\xi_{\Gamma,*}\Cont_{\Gamma},
$$
where $\xi_{\Gamma}\colon \prod_{v\in \Gamma}\Mbar_{g(v),n(v)}\to \Mbar_{g,n}$ is the canonical map with image equal to the boundary stratum associated to the graph $\Gamma$, and its push-forward $\xi_{\Gamma,*}$ induces a homomorphism from the strata algebra on  $\prod_{v\in \Gamma}\Mbar_{g(v),n(v)}$ to the cohomology ring.
Secondly, $T(\givz)
\in \givz\, \Frob[[\givz]]$, the $T$-action $T\cdot$ on a CohFT $\Omega'$ is defined by
$$\textstyle
(T\cdot\Omega')_{g,n}(-)
=\sum_{m\ge 0}\frac{1}{m!}(\pi_m)_{*}\Omega'_{g,n+m}(-\otimes T(\psi_{n+1})\otimes\cdots\otimes T(\psi_{n+m})),
$$
where $\pi_{m}\colon \overline{\mathcal{M}}_{g, n+m} \to \Mgn$ is the forgetful map forgetting the last $m$ marked points.
Furthermore, for a semisimple CohFT with vacuum $\vac(\givz)$, let $\bar {\bf 1} = \sum_\beta  \bar e_\beta$, then $T(\givz)$ is explicitly given by: 
\beq\label{eqn:T-vacuum}
T(\givz)=\givz\cdot {\bar {\bf 1}}-\givz \cdot R(\givz)^{-1}\vac(\givz).
\eeq

\subsection{CohFT associated with spectral curve}
\label{subsec:TR-CohFT-ancestors}
We recall how a semi-simple CohFT $\Omega$ is constructed from the spectral curve $\cC=(\Sigma,x,y)$ \cite{DOSS14,Eyn14}.

Around each critical point $z^\beta$, we define the local Airy coordinates $\eta^\beta=\eta^{\beta}(z)$ by
$$\textstyle
x(z)=x^\beta + \frac{1}{2}\,\eta^\beta(z)^2,
$$
where $x^\beta=x(z^\beta)$. For each $\beta \in [N]$, we introduce
$$
d\zeta^{\bar\beta}(z) :=-\mathop{\Res}_{z'=z^\beta}\,\frac{B(z',z)}{\eta^\beta(z')}.
$$
The state space associated with $\cC$ is taken to be $\Frob:=\sspan\{\bar e_\beta\}_{\beta=1}^{N}$ and the symmetric $2$-form $\eta$ is taken to be $\eta(\bar e_\beta,\bar e_\gamma)=\delta_{\beta,\gamma}$.
The $R$-matrix and $T$-vector associated with the spectral curve data $\cC$ is defined by:
\beq\label{def:EO-R}
\frac{\givz}{\sqrt{2\pi \givz}}\cdot\int_{\mathfrak L_{\gamma}}e^{-x(z)/\givz}d\zeta^{\bar\beta}(z)
\asymp e^{-x^\gamma/\givz} \cdot \eta(R(-\givz)\bar e_\gamma,\bar e_\beta),
\eeq
\beq\label{def:EO-T}
\frac{\givz}{\sqrt{2\pi \givz}}\cdot \int_{\mathfrak L_{\gamma}}e^{-x(z)/\givz}dy(z)
\asymp e^{-x^\gamma/\givz}\cdot \big(\givz-\eta(\bar e_\gamma,T(\givz))\big).
\eeq
Here $\mathfrak L_\gamma$, called the Lefschetz thimble, is a path in $\Sigma$ passing only one critical point $z^\gamma$ such that for any $z\in \mathfrak L_\gamma$, $x(z)-x(z^\gamma)\in \mathbb R_{\geq 0}$.
It is proved by Eynard~\cite{Eyn14} that  matrix $R(\givz)$ satisfies the symplectic condition: $R^{*}(-\givz)R(\givz)=\id$.
\begin{definition} Let $R$ and $T$ be defined as above.
The CohFT $\Omega$ associated with the spectral curve {$\cC$} is defined as follows:
$$
\Omega:=R\cdot T\cdot  \big(\oplus_{\beta=1}^{N}\Omega^{\rm KW_{\beta}}\big).
$$
\end{definition}
\noindent This is a semi-simple CohFT with vacuum, where the vacuum vector $\vac(\givz)$ is determined by the $R$-matrix and the $T$-vector via equation~\eqref{eqn:T-vacuum}.
Given the CohFT $\Omega$, the ancestor correlator $\<-\>_{g,n}$ and the total ancestor potential $\cA({\bf s};\hbar)$ are defined by equations~\eqref{def:bracket-anc} and \eqref{def:FandA-CohFT}, respectively.

We recall how the ancestor correlators of $\Omega$ are directly read off from the topological recursion. For each $\beta\in [N]$ and $k\geq 0$, let
$$
d\zeta_k^{\bar \beta}=\bigg(-d\circ \frac{1}{dx}\bigg)^kd\zeta^{\bar\beta},\qquad k\geq 0.
$$
The differential $d\zeta^{\bar\beta}_{k}(z)$ are globally defined meromorphic $1$-forms on $\Sigma$ with poles at the critical points.
Dunin-Barkowshi--Orantin--Shadrin--Spitz~\cite{DOSS14} (similar arguments can be also found in \cite{Eyn14, Mil14}) proved that for $2g-2+n>0$,
\beq\label{eqn:wgn-zeta}
	\omega_{g,n}(z_1,\cdots,z_n)
	=\sum_{\substack{k_1,\cdots,k_n\geqslant 0\\ \beta_1,\cdots,\beta_n\in [N]}}
	\<\bar e_{\beta_1}\bar\psi^{k_1},\cdots,\bar e_{\beta_n}\bar\psi^{k_n}\>_{g,n}
	d\zeta^{\bar\beta_1}_{k_1}(z_1)\cdots d\zeta^{\bar\beta_n}_{k_n}(z_n).
\eeq
Furthermore, it is proved in~\cite{GJZ23} that the vacuum vector associated with the spectral curve data $\cC$ has the following formula:
\beq\label{eqn:TR-vacuum}
\vac(\givz)=-\sum_{\beta}\bar e_{\beta}\sum_{m\geq 0}\givz^{m}\sum_{\gamma}\mathop{\Res}_{z= z^\gamma}  y(z)d\zeta_{m}^{\bar\beta}(z).
\eeq

\subsection{Homogeneity condition}
Let $E=E_0+\sum_{i}(1-d_i)t^i\phi_i$ be an vector field on $\Frob$, where $E_0=\sum c^i\phi_i$ for some constants $c^i$,
we call it an {\it Euler vector field} if $d_i+d_j=\delta$ for $\eta(\phi_i,\phi_j)\ne 0$ and for a constant $\delta$, called the {\it conformal dimension}, that does not depend on $i,j$.
A CohFT $\Omega$ is called {\it homogeneous} with respect to $E$, if for $2g-2+n>0$,
\begin{align*}
	&\, \textstyle \pi_{*}\Omega_{g,1+n}(E_0,v_1,\cdots,v_n)-\sum_{i=1}^{n}\Omega_{g,n}(v_1,\cdots,\deg v_i,\cdots,v_n)\nonumber\\
	=&\, \big((g-1)\delta-\deg\big)\Omega_{g,n}(v_1,\cdots,v_n).
\end{align*} 
where $\deg$ is an operator defined by
$\deg \phi_i=d_i\phi_i$ and $\deg \theta= k\cdot \theta$ for $\theta \in H^{2k}(\Mbar_{g,n}, \mathbb{A})$.
Given a homogeneous CohFT, we define operator $\mu\in \End(\Frob)$ by $\mu(\phi_i)=(d_i-\frac{\delta}{2})\phi_i$.

For semisimple CohFT, the CohFT $\Omega$ constructed by equation~\eqref{eqn:Giv-Tel} is homogeneous with respect to an Euler vector field $E$ if and only if the vacuum vector satisfies the homogeneity condition
$$\textstyle
(\mu+\frac{\delta}{2}+m)\vac_{m}=-E_0*\vac_{m+1},\qquad m\geq 0,
$$  
and the $R$-matrix satisfies the following homogeneity condition~\cite[Proposition 8.5 and Remark 8.2]{Tel12}:
\beq\label{eqn:hom-R}
[R_{m+1}, E_0*]=(m+\mu)R_{m},\qquad m\geq 0.
\eeq

Now we introduce an homogeneity condition for the spectral curve which corresponds to the homogeneity condition for the $R$-matrix on the CohFT side.
We introduce 
$$\textstyle
E_0=\sum_\beta x^\beta e_\beta.
$$
It plays the role of the Euler vector field at original point $\ttau=0$.
We define operator $\mu$ by
$$
\mu=[R_1,E_0*].
$$
We note here that by the definition and by using the airy coordinate $\eta^{\beta}=\eta^{\beta}(z)$, we have
\beq\label{eqn:R1}
(R_1)_{\bar\gamma}^{\bar\beta}=\mathop{\Res}_{z_1=z^\beta}\mathop{\Res}_{z_2=z^\gamma}\frac{B(z_1,z_2)}{\eta^\beta(z_1)\eta^\gamma(z_2)}.
\eeq
\begin{definition}
We call the spectral curve $\cC$ is homogeneous if the following equations hold:
\beq\label{eqn:Dzeta-Psi-d}
(x-x^\beta)\cdot d\zeta_1^{\bar\beta}(z)=\sum_{\gamma}\big(\mu+\tfrac{3}{2}\big)^{\bar\beta}_{\bar\gamma}\cdot d\zeta_0^{\bar\gamma}(z),\qquad \beta=1,\cdots,N.
\eeq
\end{definition}

\begin{lemma}\label{lem:hom-TR} 
The spectral curve $\cC$ is homogeneous if and only if the homogeneity condition~\eqref{eqn:hom-R} of the $R$-matrix holds.
\end{lemma}
\begin{proof}
It is clear that both sides of equation~\eqref{eqn:Dzeta-Psi-d} are meromorphic and only have poles at critical point (since $x$ is meromorphic, the left-hand side of equation~\eqref{eqn:Dzeta-Psi-d} is regular at boundaries).
Near the critical point $z=z^\gamma$, the left-hand side of equation~\eqref{eqn:Dzeta-Psi-d} has expansion
$$\textstyle
{\rm l.h.s}  = \big(x^\gamma-x^\beta+\frac{1}{2}(\eta^\gamma)^2\big) d\big(\frac{\delta_{\beta,\gamma}}{(\eta^\gamma)^{3}}+(R_1)^{\bar\beta}_{\bar\gamma}\cdot\frac{1}{\eta^{\gamma}}+{\text{regular part}}\big),
$$
and the right-hand side of equation~\eqref{eqn:Dzeta-Psi-d} has expansion
$$\textstyle
{\rm r.h.s}= d\big(\delta_{\beta,\gamma}\frac{3}{2}\cdot\frac{1}{\eta^\gamma} +\mu^{\bar\beta}_{\bar\gamma}\cdot\frac{1}{\eta^\gamma}+{\text{regular part}}\big).
$$
We see ${\rm l.h.s}-{\rm r.h.s}$ is regular if and only if $(x^\gamma-x^\beta)(R_1)^{\bar\beta}_{\bar\gamma}=\mu^{\bar\beta}_{\bar\gamma}$.
Notice that in the normalized basis $E_0*=\diag\{x^\beta\}$, this is exactly the $m=0$ case of equation \eqref{eqn:hom-R}.
Hence, the $m=0$ case of equation \eqref{eqn:hom-R} holds if and only if  there is a holomorphic differential $\omega$ on $\Sigma$ such that
$$\textstyle
(x-x^\beta)\cdot d\zeta_1^{\bar\beta}(z)=\sum_{\gamma}\big(\mu+\tfrac{3}{2}\big)^{\bar\beta}_{\bar\gamma}\cdot d\zeta_0^{\bar\gamma}(z)+\omega.
$$

Furthermore, by taking action of $(-d\cdot \frac{1}{dx})^k$, $k\geq 1$, on equation~\eqref{eqn:Dzeta-Psi-d}, we have
$$\textstyle
(x-x^\beta)d\zeta_{k+1}^{\bar\beta}(z)=\sum_{\gamma}\big(\mu+k+\tfrac{3}{2}\big)^{\bar\beta}_{\bar\gamma}\cdot d\zeta_k^{\bar\gamma}(z)+(-d\circ\frac{1}{dx})^{k}\omega.
$$
Similarly as above discussions, suppose $m=0$ case of  equation~\eqref{eqn:hom-R} holds, then the $m=k$ case of equation~\eqref{eqn:hom-R} holds if and only if $(-d\circ\frac{1}{dx})^k\omega$ is holomorphic on $\Sigma$.
We see equation~\eqref{eqn:hom-R} holds if and only if $(-d\circ\frac{1}{dx})^k\omega=0$ for $k\geq 0$, this is equivalent to $\omega=0$ (because $\frac{1}{dx}$ is singular at critical points). The Lemma is proved.
\end{proof}

\section{Ancestor Virasoro constraints for topological recursion}
\label{sec:TR-ancestor-virasoro}
In this section, we first derive the ancestor Virasoro constraints for TR-differentials (equation~\eqref{eqn:TR-virasoro}) directly from the topological recursion, thereby proving Lemma~\ref{lem:TR-res-intro}. Next, we establish the equivalence between the ancestor Virasoro constraints for the CohFT associated with the spectral curve and the equation~\eqref{eqn:TR-virasoro}, thus proving  Proposition~\ref{prop:TR-virasoro-ancestor}. As a consequence, we obtain a new direct proof for the ancestor Virasoro constraints.

\subsection{Ancestor Virasoro constraints for TR-differentials}
We derive the ancestor Virasoro constraints for the TR-differentials from the topological recursion directly.
In fact, we can prove a refined version as shown in the following Lemma.
\begin{lemma} 
\label{lem:TR-res}
For $2g-2+n+1>0$, $m\geq -1$ and $\gamma=1,\cdots,N$,  
\begin{align*}
 &\,\mathop{\Res}_{z_0=z^\gamma} x(z_0)^{m+1}\cdot y(z_0)\cdot \omega_{g,n+1}(z_0,z_{[n]})\nonumber \\
=&\,  \frac{1}{2}\mathop{\Res}_{z=z^\gamma} \frac{x(z)^{m+1}}{dx(z)}\cdot \bigg(\omega_{g-1,n+2}(z,\bar z,z_{[n]})
+\sum^{\prime}_{\substack{g_1+g_2=g\\ I\sqcup J=[n]}} \omega_{g_1,|I|+1}(z,z_{I})\omega_{g_2,|J|+1}(\bar z,z_J)\bigg).
\end{align*}
\end{lemma}
\begin{proof}
By definition, $x(z_0)^{m+1}y(z_0)$ is regular at critical points.
It is clear that $K_{\beta}(z_0,z)$ is regular (with respect to $z_0$) when $z_0$ is far away from $z^\beta$.
For each $z^\beta$, we can choose a small enough neighborhood $U^{\beta}_{\epsilon_0}=\{z_0: |z_0-z^\beta|<\epsilon_0\}$ that containing no other critical point. Take $0<\epsilon<\epsilon_0$ and let 
$C^\beta_\epsilon=\{z_0: |z_0-z^\beta|=\epsilon\}$.
By the definition of the topological recursion, we have
\begin{align*}
 &\,\mathop{\Res}_{z_0=z^\gamma} x(z_0)^{m+1}\cdot y(z_0)\cdot \omega_{g,n+1}(z_0,z_{[n]})\nonumber 
=\, \oint_{z_0\in C_{\epsilon}^{\gamma}} x(z_0)^{m+1}\cdot y(z_0)\cdot \sum_\beta \mathop{\Res}_{z=z^\beta}K_\beta(z_0,z) \cdot \mathcal{R}_{g,n}.
\end{align*}
where $\mathcal{R}_{g,n}=\mathcal{R}_{g,n}(z,\bar z,z_{[n]})$ is defined by
$$
\mathcal{R}_{g,n}(z,\bar z,z_{[n]}):=\omega_{g-1,n+2}(z,\bar z,z_{[n]})
+\sum^{\prime}_{\substack{g_1+g_2=g\\ I\sqcup J=[n]}} \omega_{g_1,|I|+1}(z,z_{I})\omega_{g_2,|J|+1}(\bar z,z_J).
$$
Now we exchange the integral $\oint_{z_0\in C_{\epsilon}^{\gamma}}$ with residue $\Res_{z=z^\beta}$.
For $\gamma\ne \beta$, these two terms can be changed directly. As the term $x(z_0)^{m+1}y(z_0)\cdot K_\beta(z_0,z)$ is regular with respect to $z_0$, the result of the integration gives $0$.
For $\gamma=\beta$, we should consider $|z-z^\gamma|<\epsilon$, then we have
$$
\oint_{z_0\in C_{\epsilon}^{\gamma}}x(z_0)^{m+1}y(z_0)\cdot K_\gamma(z_0,z)
=\frac{\int_{z'=\bar z}^{z}\Res_{z_0=z'}x(z_0)^{m+1}y(z_0)B(z_0,z')}{2(y(z)-y(\bar z))dx(z)}
=\frac{1}{2} \frac{x(z)^{m+1}}{dx(z)},
$$
where we have used the property of Bergman kernel: $\mathop{\Res}_{z_1=z_2}f(z_1)\omega_{0,2}(z_1,z_2)=df(z_2)$.
The Lemma is proved.
\end{proof}

\subsection{Ancestor Virasoro constraints for topological recursion}
Now we prove the original ancestor Virasoro constraints by establishing their equivalence with the ancestor Virasoro constraints for TR-differentials.

We first give a detailed formulation of the ancestor Virasoro constraints for the semisimple CohFT associated with the spectral curve.
Let $\E=E_0*=\diag\{x(z^\beta)\}_{\beta=1}^{N}$ and let $R(\givz)$ and $\vac(\givz)$ be the $R$-matrix and vacuum vector for the CohFT, respectively. Following the notations of~\cite{GZ25}, we introduce the operator
$$
D_{\E,\givz}:=R(\givz) \cdot \big(\E+\givz^2\pd_{\givz}+\tfrac{3}{2}\givz\big)\cdot R^{-1}(\givz).
$$
The ancestor Virasoro constraints can be formulated as follows: for $m\geq -1$,
\beq\label{eqn:virasoro-ancestor}
 L_m\cA({\bf s};\hbar)=0,
\eeq
where the ancestor Virasoro operator $L_{m}$ is given by
\begin{align*}
L_m=&\, \frac{1}{2\hbar^2}\eta(\E^{m+1}s_0,  s_0)+\frac{m+1}{16}\tr(\E^m)
+\frac{1}{2} \tr(\E^{m+1}R_1)
+\sum_{k \geq 0}\sum_{l=0}^{m+k}\sum_{\beta,\gamma=1}^{N}(C_{m})_{k-1,\bar\beta}^{l,\bar\gamma} \tilde s_k^{\bar\beta}  \frac{\pd }{\pd s_{l}^{\bar\gamma}}    \nonumber \\
&\, +\frac{\hbar^2}{2}\sum_{k=0}^{m-1}\sum_{l=0}^{m-k-1}\sum_{\beta,\gamma=1}^{N}(-1)^{k+1}( C_{m})_{-k-2,\bar\beta}^{l,\bar\gamma}\frac{\pd^2 }{\pd s_{k}^{\bar\beta}\pd s_{l}^{\bar\gamma}}.
\end{align*}
Here, $\tilde s_k^{\bar \beta}=s_k^{\bar \beta}-\vac_{k-1}^{\bar\beta}$
and for $m\geq -1$, $k,l\in \mathbb Z$, $\beta,\gamma\in [N]$, terms $(C_{m})_{k,\bar\beta}^{l,\bar\gamma}$ are defined by
$$\textstyle
D_{\E,\givz}^{m+1} \bar e_\beta \givz^k
=\sum_{l\geq k}\sum_{b=0}^{N-1}(C_{m})_{k,\bar\beta}^{l,\bar\gamma}\bar e_\gamma \givz^{l}.
$$

\begin{remark}
By Lemma~\ref{lem:hom-TR}, for homogeneous spectral curve, its corresponding $R$-matrix satisfies the homogeneity condition:
$[R_{m+1},\E]=(m+\mu)R_m$, $m\geq 0$.
This gives us that 
$$\textstyle
D_{\E,\givz}=\E+\givz^2\pd_{\givz}+(\mu+\frac{3}{2})\givz.
$$
Then the corresponding ancestor Virasoro constraints coincide with the one given in~\cite{GZ25}.
\end{remark}
\begin{proposition}[=Proposition~\ref{prop:TR-virasoro-ancestor}]
For each $m\geq -1$, the ancestor Virasoro constraint~\eqref{eqn:virasoro-ancestor} is equivalent to equation~\eqref{eqn:TR-virasoro}, and hence holds for the topological recursion.
\end{proposition}
\begin{proof}
We first prove that
\beq\label{eqn:TR-to-Virasoro-lhs}
-\sum_{\gamma}\mathop{\Res}_{z_0=z^\gamma} x(z_0)^{m+1} \cdot y(z_0)\cdot \sum_{k,\beta} \bar e_\beta \bar\psi^{k} d\zeta_k^{\bar \beta}(z_0)
=D^{m+1}_{\E,\bar \psi}\, \vac(\bar\psi).
\eeq
This can be proved as follows.
Consider $x(z)=x^{\gamma}+\frac{1}{2}(\eta^\gamma)^2$, then by equations~\eqref{def:EO-R} and \eqref{def:EO-T},
$$\textstyle
y(z)=-\sum_{i\geq 0} \frac{(T_{i+1}^{\bar\gamma}-\delta_{i,0})}{(2i+1)!!}(\eta^{\gamma})^{2i+1}+{\text{even part}},\qquad
$$
and
\beq\label{eqn:local-dzeta}
\textstyle
d\zeta_k^{\bar \beta}(z)=-\sum_{j=0}^{k} (R_j)^{\bar\beta}_{\bar \gamma}
\frac{(2k-2j+1)!!}{(\eta^\gamma)^{2k-2j+2}}d\eta^\gamma+{\text{positive part}}.
\eeq
Here the ``even part" means power series of $(\eta^\gamma)^2$ and ``regular part" means power series of $\eta^\gamma$.
We have
\begin{align*}
 &\, \mathop{\Res}_{z_0=z^\gamma} \Big(\frac{(\eta^\gamma)^{2}}{2}\Big)^l \cdot y(z_0)\cdot d\zeta_k^{\bar \beta}(z_0)
=\sum_{i+j=k-l}\frac{(2i+2l+1)!!}{(2i+1)!!2^l}
(R_j)^{\bar\beta}_{\bar \gamma} \cdot (T_{i+1}^{\bar\gamma}-\delta_{i,0}).
\end{align*}
This gives
\begin{align*}
 &\, \sum_{\gamma}\mathop{\Res}_{z_0=z^\gamma} x(z_0)^{m+1} \cdot y(z_0)\cdot \sum_{k,\beta} \bar e_\beta \bar\psi^{k} d\zeta_k^{\bar \beta}(z_0)
=R(\bar\psi)\cdot (\E+\bar\psi^2\pd_{\bar\psi}+\tfrac{3}{2}\bar\psi)^{m+1} \bar\psi^{-1}(T(\bar\psi)-\bar{\bf 1}\bar\psi ).
\end{align*}
Equation~\eqref{eqn:TR-to-Virasoro-lhs} follows from $\vac(\givz)=-R(\givz)\, \givz^{-1}\, (T(\givz)-\bar{\bf 1}\givz)$.

By using the genus expansion of the ancestor Virasoro constraints 
$$\textstyle
(L_m\cA({\bf s};\hbar))/\cA({\bf s};\hbar)=\sum_{g\geq 0}\hbar^{2g-2}\mathscr L_{g,m}({\bf s}), 
$$
it is straightforward to see that Proposition~\ref{prop:TR-virasoro-ancestor} is equivalent to
$$
\<D^{m+1}_{\E,\bar \psi}\vac(\bar\psi)\>_{1,1}
= \tfrac{1}{2}\, \tr(\E^{m+1}\cdot R_1),
$$
$$\<D^{m+1}_{\E,\bar \psi}\vac(\bar\psi),
\bar e_{\beta_{1}}\bar\psi^{k_{1}}, \bar e_{\beta_{2}}\bar\psi^{k_{2}}\>_{0,3}  d\zeta_{k_1}^{\bar \beta_1}(z_1) d\zeta_{k_2}^{\bar \beta_2}(z_2)
= \delta_{k_1,0}\delta_{k_2,0}\sum_{\gamma}(x^\gamma)^{m+1} d\zeta_0^{\bar \gamma}(z_1)d\zeta_0^{\bar\gamma}(z_2),
$$
and for $2g-2+n>0$,
\begin{align*}
&\, \<D^{m+1}_{\E,\bar \psi}\vac(\bar\psi),\bar e_{\beta_{[n]}}\bar\psi^{k_{[n]}}\>_{g,n+1}\prod d\zeta_{k_i}^{\bar \beta_i}(z_i)
=\sum_i\<[D^{m+1}_{\E,\bar\psi}\bar e_{\beta_i} \bar\psi^{k_i-1}]_{+},\bar e_{\beta_{[n]\setminus\{i\}}}\bar\psi^{k_{[n]\setminus\{i\}}}
\>_{g,n}\prod d\zeta_{k_i}^{\bar \beta_i}(z_i) \\
&\, +\frac{1}{2}\sum_{l,\sigma}(-1)^{l+1}\bigg(\<\bar e_\sigma\bar\psi^l,[D^{m+1}_{\E,\bar\psi}\bar\psi^{-1}\bar e_\sigma \bar\psi^{-l-1}]_{+},\bar e_{\beta_{[n]}}\bar\psi^{k_{[n]}}\>_{g-1,n+2}
\\
&\, \qquad 
+\sum_{\substack{g_1+g_2=g\\ I\sqcup J=[n]}} \<\bar e_\sigma\bar\psi^l,\bar e_{\beta_I}\bar\psi^{k_{I}}\>_{g_1,|I|+1}
\<[D^{m+1}_{\E,\bar\psi}\bar\psi^{-1}\bar e_\sigma \bar\psi^{-l-1}]_{+}, \bar e_{\beta_J}\bar\psi^{k_{J}}\>_{g_2,|J|+1}
\bigg)\prod d\zeta_{k_i}^{\bar \beta_i}(z_i).
\end{align*}
These equations follow from the following equalities:
\beq\label{eqn:TR-to-Virasoro-rhs-1}
-\sum_{\gamma}\mathop{\Res}_{z=z^\gamma} \frac{x(z)^{m+1}}{dx(z)}\cdot 
B(z_1, z)B(\bar z,z_2) 
=\sum_{\gamma}(x^\gamma)^{m+1} d\zeta^{\bar \gamma}(z_1)d\zeta^{\bar\gamma}(z_2),
\eeq
\beq\label{eqn:TR-to-Virasoro-rhs-2}
-\sum_{\gamma}\mathop{\Res}_{z=z^\gamma} \frac{x(z)^{m+1}}{dx(z)}\cdot 
B(z,\bar z)=\sum_{\gamma}(R_1)^{\bar\gamma}_{\bar\gamma}\cdot (x^\gamma)^{m+1}
+\frac{m+1}{8}\sum_\gamma (x^\gamma)^m,
\eeq
and
\beq\label{eqn:TR-to-Virasoro-rhs-3}
-\sum_{\gamma}\mathop{\Res}_{z=z^\gamma} \frac{x(z)^{m+1}}{dx(z)}\cdot 
d\zeta_{l}^{\bar\sigma}(z)\cdot \sum_{k,\beta} \bar e_\beta \bar\psi^k d\zeta_{k}^{\bar\beta}(\bar z)  
=(-1)^{l+1}\, [D^{m+1}_{\E,\bar\psi}\bar\psi^{-1}\bar e_\sigma \bar\psi^{-l-1}]_{+},
\eeq
\beq\label{eqn:TR-to-Virasoro-rhs-4}
-\sum_{\gamma}\mathop{\Res}_{z=z^\gamma} \frac{x(z)^{m+1}}{dx(z)}\cdot 
B(z_1,z)\sum_{k,\beta} \bar e_{\beta}\bar\psi^{k} d\zeta_{k}^{\bar\beta}(\bar z) 
=\sum_{k,\beta}[D^{m+1}_{\E,\bar\psi}\bar e_\beta \bar\psi^{k-1}d\zeta_k^{\bar\beta}(z_1)]_{+}.
\eeq
Equation~\eqref{eqn:TR-to-Virasoro-rhs-1} follows from that $B(-, z)B(\bar z,-)$ is regular, and equation~\eqref{eqn:TR-to-Virasoro-rhs-2} follows from the expansion $B(z,\bar z)=-\frac{d\eta^{\gamma}d\eta^\gamma}{4(\eta^{\gamma})^2}-(R_1)^{\bar\gamma}_{\bar\gamma}\, d\eta^{\gamma}d\eta^\gamma+O(\eta^{\gamma}) \, d\eta^{\gamma}d\eta^\gamma$.
Similarly as above discussions, recall
$$\textstyle
d\zeta_{k}^{\bar\beta}(z)|_{\gamma}=-\sum_{j\geq 0} (R_{j})^{\bar \beta}_{\bar \gamma} \frac{(2k-2j+1)!!}{\eta^{2k-2j+2}}d\eta^\gamma +({\text{odd part}})\cdot d\eta^\gamma,\qquad k \in \mathbb Z,
$$
where we have used $(-2i-1)!!=\frac{(-1)^i}{(2i-1)!!}$.
We have
$$
\mathop{\Res}_{z=z^\gamma} \frac{((\eta^\gamma)^2/2)^{n}}{dx(z)}\cdot 
d\zeta_{l}^{\bar\sigma}(z)\cdot d\zeta_{k}^{\bar\beta}(\bar z)  
=-\frac{1}{2^n}\sum_{i+j=k+l+2-n}(2l-2i+2)!!(2k-2j+2)!!(R_i)^{\bar \sigma}_{\bar\gamma} (R_j)^{\bar \beta}_{\bar\gamma},
$$
and this proves equation~\eqref{eqn:TR-to-Virasoro-rhs-3}.
By using~\cite[Lemma A.3]{GJZ23}, on the left hand side of equation~\eqref{eqn:TR-to-Virasoro-rhs-4}, one can replace $B(z_1,z)$ by
$\sum_{k,\beta}(-1)^kd\zeta_{-k-1}^{\bar\beta}(z)d\zeta_k^{\bar\beta}(z_1)$.
Then by similar discussions as above, one gets equation~\eqref{eqn:TR-to-Virasoro-rhs-4}.
The Theorem is proved.
\end{proof}

\section{Descendent Virasoro constraints for topological recursion}
\label{sec:TR-descendent-virasoro}
In this section, we begin by recalling the definition of the TR descendent invariants following~\cite{GJZ23}. We then establish the descendent Virasoro constraints for the TR differentials and, consequently, for the TR descendent potentials in terms of the KP variables $p_k^i$. 

\subsection{TR descendents }
\label{subsec:TR-descendents}

Pick $\Lambda=(\lambda_1,\cdots,\lambda_\bm)$, such that $\lambda_i$ is the local coordinate near $b_i$ satisfying $\lambda_i(b_i)={\infty}$ for each $i\in [\bm]$.
For $(i_1,\cdots,i_n)\in [\bm]^{\times n}$, $(k_1,\cdots,k_n)\in \mathbb Z_{>0}^{\times n}$ and $\alpha_k^i=k\frac{\partial}{\partial p_k^i}$,  we define   the \emph{TR descendent invariants}	
$\<\alpha^{i_1}_{k_1},\cdots,\alpha^{i_n}_{k_n}\>^{\Lambda}_{g,n}$
by taking the   expansion  of the multi-differential forms ${\omega_{g,n}}$ at the boundary points. Namely,
for $2g-2+n\geq0$, near  $z_{1}=b_{i_{1}},\cdots, z_{n}=b_{i_{n}}$  we define
$\<-\>^{\Lambda}_{g,n}$ by
\beq\label{eqn:stable-omega-gn-boundary}
\omega_{g,n}(z_1,\cdots,z_n)
=\delta_{g,0}\delta_{n,2}\frac{\delta_{i_1,i_2}d\lambda_{i_1,1}d\lambda_{i_2,2}}{(\lambda_{i_1,1}-\lambda_{i_2,2})^2}
+\sum_{k_1,\cdots, k_n\geq 1}\<\alpha^{i_1}_{k_1},\cdots,\alpha^{i_n}_{k_n}\>_{g,n}^{\Lambda } \frac{d\lambda_{i_1,1}^{-k_1}\cdots d\lambda_{i_n,n}^{-k_n}}{k_1\cdots k_n}.
\eeq
For $(g,n)=(0,1)$, we have
\beq\label{eqn:w01-boundary}
y(z)dx(z)=\sum_{k\geq 0}v^i_{k} \lambda_i^{k-1}d\lambda_i +\sum_{k\geq 1}\<\alpha_k^i\>_{0,1}^{\Lambda}\frac{d\lambda_i^{-k}}{k}.
\eeq
For the cases $(g,n) = (0,0), (1,0)$, all the invariants $\<-\>^{\Lambda}_{g,n}$ are taken to be zero.
We define the generating series of TR descendents:
\beq\label{eqn:generating-series-m-KP-F-A-TR}
Z^{\Lambda}({\bf p};\hbar)=\exp\bigg(\sum_{g\geq 0, n\geq 0}\hbar^{2g-2}\sum_{1\leq i_1,\cdots,i_n\leq \bm \atop k_1,\cdots,k_n\geq 1}
\<\alpha^{i_1}_{k_1},\cdots,\alpha^{i_n}_{k_n}\>_{g,n}^{\Lambda} \frac{p^{i_1}_{k_1}\cdots p^{i_n}_{k_n}}{n!\cdot k_1\cdots k_n}\bigg).
\eeq

We have the following formula of the generating series:
\beq\label{eqn:DKP-ACohFT}
Z^{\Lambda}({\bf p};\hbar)=e^{\frac{1}{\hbar^2}J_{-}({\bf p})+\frac{1}{2\hbar^2}Q({\bf p},{\bf p})}\cdot\cA({\bf s(p)};\hbar),
\eeq
where $J_{-}({\bf p})$ and $Q({\bf p},{\bf p})$ are two functions defined by
$$
J_{-}({\bf p}):=\sum_{k>0;\, i\in [\bm]}\<\alpha_k^i\>_{0,1}^{\Lambda}\frac{p_k^i}{k},\qquad
Q({\bf p},{\bf p}):=\sum_{k,l>0;\, i,j\in [\bm]}\<\alpha^i_k,\alpha^j_{l}\>_{0,2}^{\Lambda}\frac{p^i_k}{k}\frac{p^j_{l}}{l},
$$
and ${\bf s(p)}$ is the coordinate transformation determined by the local expansion of $d\zeta_k^{\bar\beta}(z)$ at boundary points:
\beq\label{eqn:s-p}\textstyle
d\zeta_m^{\bar\beta}(z)|_{b_i}= \sum_{k} c^{m,\bar\beta}_{i,k} d\lambda_i^{-k}
\qquad\Longrightarrow \qquad
s_m^{\bar\beta}({\bf p})=\sum_{k,i}c^{m,\bar\beta}_{i,k}\, p^i_k.
\eeq

\subsection{TR descendent Virasoro constraints}
We consider the cases that the functions $x$, $y$ in spectral curve data $\cC$ are meromorphic on the Riemann surface $\Sigma$.
To prove the TR descendent Virasoro constraints, we first prove the following equivalent statement in terms of multi-differentials $\omega_{g,n}$, which we refer to as \emph{descendent Virasoro constraints for TR-differentials}.
\begin{proposition}[= Proposition~\ref{prop:TR-crit-boundary-intro}] 
\label{prop:TR-crit-boundary}
For fixed $m\geq -1$, if $x$, $y$ are meromorphic and satisfy that $x^{m+1}y$ has only poles at the boundary points, then for $2g-2+n+1>0$, we have
\begin{align}
 &\, -\sum_{i}\mathop{\Res}_{z_0=b_i} x(z_0)^{m+1}\cdot y(z_0)\cdot \omega_{g,n+1}(z_0,z_{[n]})
\nonumber \\
=&\,  \frac{1}{2}\sum_{i}\mathop{\Res}_{z=b_i} \frac{x(z)^{m+1}}{dx(z)}\cdot \bigg(\tilde\omega_{g-1,n+2}(z, z,z_{[n]})
+\sum^{\prime}_{\substack{g_1+g_2=g\\ I\sqcup J=[n]}} \omega_{g_1,|I|+1}(z,z_{I})\omega_{g_2,|J|+1}(z,z_J)\bigg).
\label{eqn:TR-virasoro-des}
\end{align}
where $\tilde\omega_{g,n}(z_{[n]})=\omega_{g,n}(z_{[n]})-\delta_{(g,n),(0,2)}\cdot\frac{dx(z_1)dx(z_2)}{(x(z_1)-x(z_2))^2}$.
\end{proposition}
\begin{proof}
By assumption and the residue formula, the left-hand side of equation~\eqref{eqn:TR-virasoro-des} coincides with that of equation~\eqref{eqn:TR-virasoro}. It thus suffices to prove that their right-hand sides are equal.

For $2g - 2 + n > 0$, using the structure of $\omega_{g,n}$ (equation~\eqref{eqn:wgn-zeta}), the local expansion of $d\zeta_k^{\bar{\beta}}(z)$ at the critical point $z^\gamma$ (equation~\eqref{eqn:local-dzeta}), and the identity $\bar{z}(\eta^\gamma) = z(-\eta^\gamma)$, we see that the right-hand side of equation~\eqref{eqn:TR-virasoro} equals
$$
-\frac{1}{2}\sum_{\gamma}\mathop{\Res}_{z=z^\gamma} \frac{x(z)^{m+1}}{dx(z)}\cdot \bigg(\omega_{g-1,n+2}(z, z,z_{[n]})
+\sum^{\prime}_{\substack{g_1+g_2=g\\ I\sqcup J=[n]}} \omega_{g_1,|I|+1}(z,z_{I})\omega_{g_2,|J|+1}( z,z_J)\bigg).
$$
This further equals the right-hand side of~\eqref{eqn:TR-virasoro-des} by the residue formula.

For $(g,n)=(1,0)$, the right hand side of equation~\eqref{eqn:TR-virasoro} is given by
$$
\frac{1}{2}\sum_{\gamma}\mathop{\Res}_{z^{\gamma}}\frac{x(z)^{m+1}}{dx(z)}B(z,\bar z).
$$
Notice that near the point $z^\gamma$, one has
$$
\frac{d\eta^\gamma_1d\eta^\gamma_2}{(\eta^\gamma_1-\eta^\gamma_2)^2}-\frac{d(\eta^\gamma_1)^2d(\eta^\gamma_2)^2}{((\eta^\gamma_1)^2-(\eta^\gamma_2)^2)^2}
=\frac{d\eta^\gamma_1d\eta^\gamma_2}{(\eta^\gamma_1+\eta^\gamma_2)^2}.
$$
By comparing
$$
B(z,\bar z)|_{z^\gamma} =-\frac{d\eta^\gamma_1d\eta^\gamma_2}{(\eta^\gamma_1+\eta^\gamma_2)^2}-B^{\gamma,\gamma}_{0,0}d\eta^\gamma_1d\eta^\gamma_2+\cdots
$$
with
$$
\Big(\frac{dx(z_1)dx(z_2)}{(x(z_1)-x(z_2))^2}-B(z_1,z_2)\Big)\Big|_{(z^\gamma,z^\gamma)}
=-\frac{d\eta^\gamma_1d\eta^\gamma_2}{(\eta^\gamma_1+\eta^\gamma_2)^2}-B^{\gamma,\gamma}_{0,0}d\eta^\gamma_1d\eta^\gamma_2+\cdots.
$$
we have
$$
\frac{1}{2}\sum_{\gamma}\mathop{\Res}_{z=z^{\gamma}}\frac{x(z)^{m+1}}{dx(z)}B(z,\bar z)
=-\frac{1}{2}\sum_{\gamma}\mathop{\Res}_{z=z^{\gamma}}\frac{x(z)^{m+1}}{dx(z)}\tilde \omega_{0,2}(z,z).
$$
Notice that $\frac{x(z)^{m+1}}{dx(z)}\tilde \omega_{0,2}(z,z)$ is a global-defined meromorphic one form on $\Sigma$ and has only possible poles at the critical points (of $x(z)$) and boundary points, the residue formula gives us that 
$$
-\frac{1}{2}\sum_{\gamma}\mathop{\Res}_{z=z^{\gamma}}\frac{x(z)^{m+1}}{dx(z)}\tilde \omega_{0,2}(z,z)
=\frac{1}{2}\sum_{i}\mathop{\Res}_{z=b_i}\frac{x(z)^{m+1}}{dx(z)}\tilde \omega_{0,2}(z,z).
$$
The Proposition is proved.
\end{proof}

Now we take a special choice of the local coordinates $\Lambda$ of the boundaries as follows: for each $i=1,\cdots,\bm$, near the boundary $b_i$, $x$ can be viewed as a $r_i:1$ covering map,
we define the corresponding local coordinate $\lambda_i$, up to a factor $\xi_i$ such that $\xi_i^{r_i}=1$, by
$$
x(z)|_{b_i}=\lambda_i^{r_i}.
$$
We denote the corresponding TR descendent generating series by $Z({\bf p};\hbar)$.
For further use, we introduce the following notations $p_0^i=0$ and
$$
\tilde p_k^{i}=p_k^i-v_k^i,\qquad k\geq 0.
$$
We note that $\tilde p_0^i = -v_0^i$ is not necessarily zero.

\begin{theorem}[=Theorem~\ref{thm:TR-virasoro-descendent-intro}]\label{thm:TR-virasoro-descendent}
For fixed $m\geq -1$, if the meromorphic functions $x$ and $y$ satisfy that the function $x^{m+1}y$ has only poles at the boundary points, then the generating series $Z({\bf p};\hbar)$ satisfies the following equation:
$$
\cL_mZ({\bf p};\hbar)=0.
$$
Here the operator $\cL_m$ is defined by
\begin{align*}
\cL_m=&\, \frac{c_m}{2\hbar^2}+
\frac{\delta_{m,-1}}{2\hbar^2}\sum_{i}\frac{1}{r_i} \sum_{a=0}^{r_i}\tilde p^i_{a}\tilde p^i_{r_i-a}
+\frac{\delta_{m,0}}{2\hbar^2}\sum_{i}\frac{(v_0^i)^2}{r_i} 
+\delta_{m,0}\sum_{i}\frac{r_i^2-1}{24r_i} \\
&\, +\sum_{i}\frac{1}{r_i}\sum _{k\geq 0} \tilde p^i_k (r_im+k)\frac{\pd}{\pd p^i_{r_im+k}} 
+\frac{\hbar^2}{2}\sum_{i}\frac{1}{r_i}\sum_{k+l=r_im}kl\frac{\pd^2}{\pd p^i_k\pd p^i_l},
\end{align*}
where the summation $\sum_i$ is over the indices labeling the boundary points, and the constant $c_m$ is given by
$$
c_m=-\sum_{z^0:\,  x(z^0)=0} \mathop{\Res}_{z=z^0} x(z)^{m+1}y(z)\omega_{0,1}(z).
$$
\end{theorem}
\begin{proof}
We identify $\alpha_k^i$ with $k\frac{\pd}{\pd p_k^i}$, and for $k\geq 0$, we define $\<\alpha^i_{-k},-\>_{g,n}=0$ unless 
$$
\<\alpha_{-k}^{i}\>_{0,1}:=-v_k^i,\qquad
\<\alpha_{-k}^i,\alpha_{l}^i\>_{0,2}=\<\alpha_{l}^i,\alpha_{-k}^i\>_{0,2}=\delta_{k,l}\cdot k.
$$
By exchanging the operator $\cL_m$ with $e^{\frac{1}{\hbar^2}J_{-}({\bf p})}$, we have
$$
\cL_m e^{\frac{1}{\hbar^2}J_{-}({\bf p})}
=e^{\frac{1}{\hbar^2}J_{-}({\bf p})}\widetilde\cL_{m},
$$
where
\begin{align*}
\widetilde\cL_{m}=&\, \frac{{\rm const}}{2\hbar^2}+\delta_{m,0}\sum_{i}\frac{r_i^2-1}{24r_i}
+\frac{1}{\hbar^2}\sum_i\frac{1}{r_i}\sum_{k\geq 1}\<\alpha_{r_im+k}^i\>_{0,1}\, p^i_{k}
+\sum_{i}\frac{1}{r_i}\sum_{k\geq 1} \<\alpha_{r_im-k}^{i}\>_{0,1}\cdot \alpha_k^i\\
&\, +\frac{\delta_{m,-1}}{2\hbar^2}\sum_{i}\frac{1}{r_i} \sum_{a=0}^{r_i} p^i_{a} p^i_{r_i-a}
+\sum_{i}\frac{1}{r_i}\sum _{k\geq 0}  p^i_k \alpha^i_{k+r_im}
+\frac{\hbar^2}{2}\sum_{i}\frac{1}{r_i}\sum_{k+l=r_im}\alpha^i_k\alpha^i_l.
\end{align*} 
Here the term ``const" is given by
$$
c_m+\delta_{m,-1}\sum_i\frac{v_{a}^{i}v_{r_i-a}^i}{r_i}\
+\delta_{m,0}\sum_i\frac{(v_0^i)^2}{r_i}
-2\sum_i\frac{1}{r_i} v_k^i\<\alpha^i_{k+r_im}\>_{0,1}
+\sum_i\frac{1}{r_i}\sum_{k+l=r_im}\<\alpha^i_k\>\<\alpha^i_l\>.
$$
Notice that $\frac{dx(z)}{x(x)}=r_i \frac{d\lambda_i}{\lambda_{i}}$, we have
$x(z)y(z)|_{b_i}=\frac{1}{r_i}\sum_{k\geq 0} v_k^i \lambda_i^{k} -\frac{1}{r_i}\sum_{k\geq 1}\<\alpha^i_k\>_{0,1}\lambda_i^{-k}$,
and
$$
x(z)^2y(z)^2|_{b_i}=
\frac{1}{r_i^2}\sum_{k,l\geq 0} v_k^i v^i_l \lambda_i^{k+l}
-\frac{2}{r_i^2}\sum_{k\geq 1,l\geq 0}v_l^i\<\alpha^i_k\>_{0,1}\lambda_i^{l-k}
+\frac{1}{r_i^2}\sum_{k,l\geq 1} \<\alpha^i_k\>_{0,1}\<\alpha^i_l\>_{0,1} \lambda_i^{-l-k}.
$$
Then it is easy to see 
$$
{\rm const}=
-\bigg(\sum_{z^0:\,  x(z^0)=0}\mathop{\Res}_{z=z^0}+\sum_{z^0:\,  x(z^0)=\infty}\mathop{\Res}_{z=z^0}\bigg) x(z)^{m+1}y(z)\omega_{0,1}(z)=0.
$$

Now we prove
$\widetilde \cL_m \big(e^{-\frac{1}{\hbar^2}J_{-}({\bf p})}Z({\bf p};\hbar)\big)=0$.
By taking genus expansion, and by noticing
$$
\delta_{m,-1}\delta_{g,0}\delta_{n,2}\delta_{j_1,j_2}\delta_{k_1+k_2=r_{j_a}} \frac{k_1k_2}{r_{j_1}}
=\sum_i\frac{1}{r_i}\sum_{k+l=r_i, k,l\geq 0}\<\alpha^i_{-k},\alpha^{j_1}_{k_1}\>_{0,2} \<\alpha^i_{-l},\alpha^{j_2}_{k_2}\>_{0,2},
$$
and  
$$
\frac{k_a}{r_{j_a}}\<\alpha^{j_a}_{k_a+r_{j_a}m},\alpha^{j_{[n]\setminus \{a\}}}_{k_{[n]\setminus \{a\}}}\>_{g,n}
=\sum_{i}\frac{1}{r_i}\sum_{k\geq 0}\<\alpha_{-k}^{i},\alpha_{k_a}^{j_a}\>_{0,2}\<\alpha^{i}_{k+r_{i}m},\alpha^{j_{[n]\setminus \{a\}}}_{k_{[n]\setminus \{a\}}}\>_{g,n},
$$
the equation is equivalent to the following statement:
for all $(g,n)$ satisfying $2g - 2 + n + 1 > 0$, and for each collection ${(k_a, j_a)}_{a=1}^n$, the expression
\begin{align}
&\, \delta_{m,0}\delta_{g,1}\delta_{n,0}\sum_{i}\frac{r_i^2-1}{24r_i}
+\sum_{i}\frac{1}{r_i}\sum_{k\in\mathbb Z} \<\alpha_{r_im-k}^{i}\>_{0,1}\cdot \<\alpha_k^i,\alpha^{j_{[n]}}_{k_{[n]}}\>_{g,1+n}\nonumber\\
&\, +\frac{1}{2}\sum_{i}\frac{1}{r_i}\sum_{k+l=r_im}
\bigg(\<\alpha^i_k,\alpha^i_l,\alpha^{j_{[n]}}_{k_{[n]}}\>_{g-1,n+2}
+\sum^{\prime}_{\substack{g_1+g_2=g\\ I\sqcup J=[n]}}
\<\alpha^i_k,\alpha^{j_{I}}_{k_{I}}\>_{g_1,|I|+1}\<\alpha^i_l,\alpha^{j_{J}}_{k_{J}}\>_{g_2,|J|+1}\bigg)
\label{eqn:TR-virasoro-des-corr}
\end{align}
equals zero.
Here, in the term $\<\alpha^i_k,\alpha^i_l,\alpha^{j_{[n]}}_{k_{[n]}}\>_{g-1,n+2}$, the subscripts  $k$ and $l$ denote positive integers.
For $(g,n) \neq (1,0)$, by multiplying expression~\eqref{eqn:TR-virasoro-des-corr} by $\prod_a d\lambda_{j_a}^{-k_a}$ and summing over all $k_a \geq 1$, it is straightforward to see that the vanishing of expression~\eqref{eqn:TR-virasoro-des-corr} is equivalent to equation~\eqref{eqn:TR-virasoro-des} holding at the boundary points $(b_{j_1},\cdots,b_{j_n})$.
This is thus ensured by Proposition~\ref{prop:TR-crit-boundary}.

It remains to prove the case with $(g,n)=(1,0)$, for which case the equation is
\beq\label{eqn:virasoro-01}
\delta_{m,0}\sum_{i}\frac{r_i^2-1}{24r_i}
+\sum_{i}\frac{1}{r_i}\sum_{k\geq 1} \<\alpha_{r_im-k}^{i}\>_{0,1}\cdot \<\alpha_k^i\>_{1,1}
+\frac{1}{2}\sum_{i}\frac{1}{r_i}\sum_{k+l=r_im;\, k,l\geq 1}
\<\alpha^i_k,\alpha^i_l\>_{0,2}
=0.
\eeq
Notice that 
$$
\sum_{i}\frac{1}{r_i}\sum_{k\geq 1} \<\alpha_{r_im-k}^{i}\>_{0,1}\cdot \<\alpha_k^i\>_{1,1}
=-\sum_i\mathop{\Res}_{z=b_i}x(z)^{m+1}y(z)\omega_{1,1}(z),
$$ 
and by Proposition~\ref{prop:TR-crit-boundary}, this further equals
\beq\label{eqn:res-tilde-omega02}
\frac{1}{2}\sum_{i}\mathop{\Res}_{z=b_i}\frac{x(z)^{m+1}}{dx(z)}\tilde\omega_{0,2}(z,z).
\eeq
Near each boundary $b_i$, we use local coordinate $X_i=\lambda_{i}^{-1}$, then $x(z)=X_i^{-r_i}$, and we have
$$
\tilde \omega_{0,2}(z,z)=\bigg(\frac{dX_{i}dX'_{i}}{(X_{i}-X'_{i})^2}-\frac{dX_i^{-r_i}d(X'_i)^{-r_i}}{(X_i^{-r_i}-(X'_i)^{-r_i})^2}\bigg)\bigg|_{X'_i=X_i}+\sum_{k,l\geq 1}\<\alpha^i_k,\alpha^i_l\>_{0,2}\cdot X_i^{k+l}\cdot\frac{dX_i}{X_i}\cdot\frac{dX_i}{X_i}.
$$
By multiplying by $\frac{x(z)^{m+1}}{dx(z)} = -\frac{1}{r_i} \cdot \frac{1}{X_i^{r_i m}} \cdot \frac{X_i}{dX_i}$ and then taking the residue at $X_i = 0$,
the regular part of $\tilde\omega_{0,2}(z,z)$ in equation~\eqref{eqn:res-tilde-omega02} contributes
$$
-\frac{1}{2}\sum_{i}\frac{1}{r_i}\sum_{k+l=r_im;\, k,l\geq 1}
\<\alpha^i_k,\alpha^i_l\>_{0,2}.
$$
To compute the contribution of the singular part of $\tilde \omega_{0,2}(z,z)$, we need to compute
$$
\frac{1}{2}\mathop{\Res}_{X_i=0}\Big(\frac{x(z)^{m+1}}{dx(z)}
\frac{(\sum_{a+b=r}X_i^a(X'_i)^b)^2-r_i^2X_i^{r-1}(X'_i)^{r-1}}{(\sum_{a+b=r}X_i^a(X'_i)^b)^2(X_i-X'_i)^2}dX_idX'_i\Big|_{X'_i=X_i}\Big).
$$
Notice
$$
\lim_{X'_i=X_i}\frac{(\sum_{a+b=r}X_i^a(X'_i)^b)^2-r_i^2X_i^{r-1}(X'_i)^{r-1}}{(\sum_{a+b=r}X_i^a(X'_i)^b)^2(X_i-X'_i)^2}
=\frac{r^2-1}{12X_i^2},
$$
we see the residue gives
$$
-\frac{1}{2}\mathop{\Res}_{X_i=0} \frac{r_i^2-1}{12X_i^2}\cdot \frac{X_i}{r_i X_i^{mr_i}}\cdot dX_i=-\delta_{m,0}\frac{r_i^2-1}{24\, r_i}.
$$
This proves equation~\eqref{eqn:virasoro-01}.
The proof is finished.
\end{proof}

\section{Virasoro constraints for non-perturbative topological recursion}
\label{sec:NP-TR-virasoro}
In this section, we recall the non-perturbative generating series~\cite{GJZ23} (see also \cite{EM11,BE15,ABDKS25}) and prove the Virasoro constraints for the non-perturbative descendent potentials.

\subsection{Non-perturbative generating series}
The non-perturbative invariants come from the non-vanished $B$-cycle periods of the multi-differentials $\omega_{g,n}$.
We first introduce the object $\omega_{g,n_1}^{(n_2)}$, an $n_1$-differential valued $n_2$-tensor, whose $(j_1,\cdots,j_{n_2})$-component  is given by
$$
\omega_{g,n_1,j_1,\cdots,j_{n_2}}^{(n_2)}(z_{[n_1]}):=\oint_{z_{n_1+1}\in \cyB_{j_1}}\!  \cdots \oint_{z_{n_1+n_2}\in \cyB_{j_{n_2}}} \omega_{g,n_1+n_2}(z_1,\cdots,z_{n_1+n_2}).
$$
Recall the structure of $\omega_{g,n}$ (equation~\eqref{eqn:wgn-zeta}), for $2g-2+n_1+n_2>0$, we have the following structure of $\omega_{g,n_1}^{(n_2)}$:
$$
\omega_{g,n_1}^{(n_2)}(z_{[n_1]})=\sum_{\substack{k_1,\cdots,k_{n_1}\geqslant 0\\ \beta_1,\cdots,\beta_{n_2}\in [N]}}
\<\bar e_{\beta_1}\bar\psi^{k_1},\cdots,\bar e_{\beta_{n_1}}\bar\psi^{k_{n_1}}\>^{(n_2),\, \Omega}_{g,n_1}
d\zeta^{\bar\beta_1}_{k_1}(z_1)\cdots d\zeta^{\bar\beta_{n_1}}_{k_{n_1}}(z_{n_1}).
$$
Similarly, taking a local coordinate system $\Lambda = (\lambda_1, \dots, \lambda_{\bm})$ near the boundary points, we obtain the expansion of $\omega_{g,n_1}^{(n_2)}$ around $z_1=b_{i_1},\cdots,z_{n_1}=b_{i_{n_1}}$:
$$
\omega_{g,n_1}^{(n_2)}(z_{[n_1]})=\sum_{k_1,\cdots, k_{n_1}\geq 1}\<\alpha^{i_1}_{k_1},\cdots,\alpha^{i_{n_1}}_{k_{n_1}}\>_{g,n_1}^{(n_2), \, \Lambda}
\cdot \frac{d\lambda_{i_1,1}^{-k_1}\cdots d\lambda_{i_{n_1},n_1}^{-k_{n_1}}}{k_1\cdots k_{n_1}}.
$$
Here the terms $\<-\>_{g,n_1}^{(n_2),\Omega}$ and $\<-\>_{g,n_1}^{(n_2),\Lambda}$ are tensor-valued invariants, whose components are denoted by $\<-\>_{g,n_1,j_1,\cdots,j_{n_2}}^{(n_2),\Omega}$ and $\<-\>_{g,n_1,j_1,\cdots,j_{n_2}}^{(n_2),\Lambda}$, respectively.
For the cases $\omega_{0,n_1}^{(n_2)}$, $n_1+n_2=2$, the 2-form $\omega_{0,2}^{(0)}=\omega_{0,2}$ is expanded by equation~\eqref{eqn:stable-omega-gn-boundary} ($(g,n)=(0,2)$ case),
the $1$-form $\omega^{(1)}_{0,1,j}(z)$ near the boundary point $z=b_i$ has expansion
$$
\omega^{(1)}_{0,1,j}(z)=\oint_{z'\in\cyB_j}B(z,z')=2\pi {\bf i}\, du_j(z)
=\sum_{k\geq 1}\<\alpha_k^{i}\>_{0,1,j}^{(1),\Lambda}\frac{d\lambda_{i}^{-k}}{k},
$$
and the $0$-form $\omega_{0,0,i,j}^{(2)}$ is just a constant given by
$$
\omega_{0,0,i,j}^{(2)}=2\pi {\bf i}\oint_{z\in \cyB_j}du_i(z)=2\pi {\bf i}\, \tau_{ij}=\<\>_{0,0,i,j}^{(2),\Lambda}.
$$

Given the non-perturbative invariants, we use the theta function $\theta[\cmn](w|\tau')$ to encode all $B$-cycle integrals and define the corresponding TR non-perturbative generating series.
\begin{definition}
	Given the spectral curve $\cC=(\Sigma,x, y)$,
	the non-perturbative total ancestor potential $\cA_{\mu,\nu,\tau'}^{\rm NP}({\bf s};\hbar;w)$ is defined by
	$$
	\cA^{\rm NP}_{\mu,\nu,\tau'}({\bf s};\hbar;w):=
	\exp\bigg(\sum_{2g-2+n_1+n_2>0 } \hbar^{2g-2+n_2}\frac{\<{\bf s}(\bar\psi),\cdots,{\bf s}(\bar\psi)\>_{g,n_1}^{(n_2),\Omega} \cdot \nabla_w^{\otimes n_2}}{n_1!n_2!}\bigg)\theta[\cmn](w|\tau'),
	$$
	and the non-perturbative TR descendent potential $Z_{\mu,\nu,\tau'}^{\Lambda,\rm NP}({\bf p};\hbar;w)$ is defined by
	$$
	Z^{\Lambda,\rm NP}_{\mu,\nu,\tau'}({\bf p};\hbar;w):=
	\exp\bigg(\sum_{g,n_1,n_2\geq 0} \hbar^{2g-2+n_2}\frac{\<{\bf p}(\alpha),\cdots,{\bf p}(\alpha)\>_{g,n_1}^{(n_2), \Lambda} \cdot \nabla_w^{\otimes n_2}}{n_1!n_2!}\bigg)\theta[\cmn](w|\tau').
	$$
	Here ${\bf s}(\bar\psi)=\sum_{k\geq 0}\sum_{a=1}^{N}s_k^{\bar \beta}\bar e_\beta\bar\psi^k$, ${\bf p}(\alpha)=\sum_{k\geq 1}\sum_{i=1}^{\bm}\frac{1}{k}p_k^i \alpha_k^i $,
	 $\nabla_{w}=\frac{1}{2\pi {\bf i}}(\pd_{w_1},\cdots,\pd_{w_{\frak g}})$ and
	 $$\<-\>_{g,n_1}^{(n_2),\star}\cdot \nabla_{w}^{\otimes n_2}
	 =\frac{1}{(2\pi {\bf i})^{n_2}}\sum_{j_1,\cdots,j_{n_2}=1}^{\frak g}\<-\>_{g,n_1,j_1,\cdots,j_{n_2}}^{(n_2),*}\pd_{w_{j_1}}\cdots\pd_{w_{j_{n_2}}},
	 $$
where $\star=\Omega, \Lambda$.
\end{definition}
We note that one cannot take $\tau'=0$ in $\cA^{\rm NP}_{\mu,\nu,\tau'}({\bf s};\hbar;w)$, as this would lead to a divergent result in the theta function.
However, one can take $\tau'=0$  in $Z^{\Lambda,\rm NP}_{\mu,\nu,\tau'}({\bf p};\hbar;w)$, as explained in~\cite{GJZ23},
because the operator $e^{\frac{1}{2}\omega_{0,0}^{(2)}\nabla_{w}^{\otimes 2}}$ shifted $\tau'$ in $\theta[\cmn](w|\tau')$ by $\tau$:
 $$
 e^{\frac{1}{2}\omega_{0,0}^{(2)}\nabla_{w}^{\otimes 2}}\big(\theta[\cmn](w|\tau')\big)=\theta[\cmn](w|\tau'+\tau).
 $$
We define
$$
Z^{\Lambda,\rm NP}_{\mu,\nu}({\bf p};\hbar;w):=Z^{\Lambda,\rm NP}_{\mu,\nu,\tau'}({\bf p};\hbar;w)|_{\tau'=0}.
$$
Furthermore, we have the following non-perturbative analogy of the formula~\eqref{eqn:DKP-ACohFT}:
$$
Z^{\Lambda,\rm NP}_{\mu,\nu}({\bf p};\hbar;w):=e^{\frac{1}{2\hbar^2}J_{-}({\bf p})+\frac{1}{2\hbar^2}Q({\bf p,p})}\cdot e^{\frac{1}{\hbar}\cdot 2\pi {\bf i}\, \widehat{u(z)}\cdot \nabla_{w}}\big(\cA_{\mu,\nu,\tau}^{\rm NP}( {\bf s(p)};\hbar;w)\big),
$$
where $\widehat{u(z)}=\frac{1}{2\pi {\bf i}}\sum_{k,i}\<\alpha^i_{k}\>_{0,1}^{(1),\Lambda}\, \frac{p^i_k}{k}$ and ${\bf s(p)}$ is given by~\eqref{eqn:s-p}.

\subsection{Non-perturbative TR Virasoro constraints}
Now we consider the Virasoro constraints for non-perturbative generating series.

For the ancestor side, introduce a sequence of vector fields $\sv_j$, $j=1,\cdots,\frak g$:
\beq\label{def:filling-fraction}
\sv_j=\sum_{\beta}\bar e_{\beta}\cdot \oint_{z\in \cyB_j}d\zeta^{\bar\beta}(z)
=-2\pi {\bf i}\sum_{\beta}\frac{du_j}{dy}\Big|_{z=z^\beta}\cdot e_{\beta}.
\eeq
then it is shown in~\cite[\S 4]{GJZ23} that
$$
\cA^{\rm NP}_{\mu,\nu,\tau}({\bf s};\hbar;w)=e^{\hbar\varphi\cdot \nabla_w}(\cA({\bf s};\hbar)\cdot \theta[\cmn](w|\tau)),
$$
where $\varphi\cdot\nabla_w=\frac{1}{2\pi {\bf i}}\sum_j\varphi_j\nabla_{w_j}$.
By this formula, it is easy to see
$$
 L^{\rm NP}_m\cA^{\rm NP}_{\mu,\nu,\tau}({\bf s};\hbar;w)=0,
$$
where
$$
 L^{\rm NP}_m= L_m|_{s_0+\varphi\cdot \nabla_w}
$$

For descendent side, similar as the perturbative side, we consider cases satisfying $x$ and $y$ are meromorphic functions, and we take a special choice of local coordinate $\lambda_i$ by $x(z)=\lambda_i^{r_i}$.
For such case, we denote the corresponding generating series  $Z^{\Lambda,\rm NP}_{\mu,\nu}({\bf p};\hbar;w)$ by $Z^{\rm NP}_{\mu,\nu}({\bf p};\hbar;w)$.

\begin{theorem}[=Theorem~\ref{thm:NP-TR-virasoro-descendent-intro}]
For $m\geq -1$,  if the meromorphic functions $x$ and $y$ satisfy that the function $x^{m+1}y$ has only poles at the boundary points, then the non-perturbative generating series $Z^{\rm NP}_{\mu,\nu}({\bf p};\hbar;w)$ satisfies the following equation
$$
\cL_mZ^{\rm NP}_{\mu,\nu}({\bf p};\hbar;w)=0,
$$
where the operator $\cL_m$ is the same one as given in Theorem~\ref{thm:TR-virasoro-descendent}.
\end{theorem}
\begin{proof}
Clearly, we just need to prove $\cL_mZ^{\rm NP}_{\mu,\nu,\tau'}({\bf p};\hbar;w)=0$ for arbitrary $\tau'$. 
In fact, the proof of this equation is parallel to the perturbative version: 
just change the correlators $\<-\>_{g,n}$ in the proof of Theorem~\ref{thm:TR-virasoro-descendent} to the non-perturbative ones. 
We omit the details here.
\end{proof}

\section{Examples}
\label{sec:example}
Recall that in the first paper~\cite{GZ25}, we constructed formal descendent generating series for $S$- and $\dvac$-calibrated homogeneous CohFTs with vacuum, and proved the geometric descendent Virasoro conjecture in the semisimple case.
In this section, we present several examples of topological recursion whose underlying CohFTs are homogeneous.
We provide explicit formulae for the TR descendent Virasoro constraints and the geometric Virasoro constraints, elucidating   their relationship.
\subsection{Airy curve}
The basic example is the Airy curve:
$$
\cC^{\rm Ai}=(\mathbb P^1,\quad x=\tfrac{1}{2}z^2,\quad y=z).
$$
The corresponding CohFT is the trivial CohFT, and the corresponding geometric total descendent potential $\cD({\bf t};\hbar)$ is the generating series of the intersection theory on the Deligne--Mumford moduli space $\Mbar_{g,n}$. The corresponding TR descendent potential $Z({\bf p};\hbar)$ with respect to the local coordinate $\lambda=\sqrt{2x(z)}=z$ near $z=\infty$ (we slightly modify the general setting of $\lambda$ by a factor $\sqrt{2}$) is known as the Witten--Kontsevich tau-function.
One has $\cD({\bf t(p)};\hbar)=Z({\bf p};\hbar)$ via the coordinate transformation $t_k=(2k-1)!! p_{2k+1}$, $k\geq 0$.

Let $\tilde t_k=t_k-\delta_{k,1}$ and $\tilde p_{k}=p_k-\delta_{k,3}$, the geometric and TR Virasoro constraints are given as follows:
$$
L_m\cD({\bf t};\hbar)=\cL_mZ({\bf p};\hbar)=0,\qquad m\geq -1,
$$
where the geometric Virasoro operators $L_m$ have formulae:
$$
L_m=\frac{\delta_{m,-1}}{2\hbar^2}(t_0)^2+\frac{\delta_{m,0}}{16}
+\sum_{k\geq 0}\frac{(2k+2m+1)!!}{2^{m+1}(2k-1)!!}\tilde t_k\frac{\pd}{\pd t_{k+m}}+\frac{\hbar^2}{2}\sum_{k+l=m-1}\frac{(2k+1)!!(2l+1)!!}{2^{m+1}}\frac{\pd^2}{\pd t_k\pd t_l}.
$$
and the TR Virasoro operators $\cL_m$ have formulae:
\begin{align*}
\cL_m=&\, \frac{\delta_{m,-1}}{4\hbar^2}(p_1)^2+\frac{\delta_{m,0}}{16}
+\sum_{k\geq 0}\frac{2k+2m+1}{2^{m+1}}\tilde p_{2k+1}\frac{\pd }{\pd p_{2k+2m+1}}\\ 
&\, +\frac{\hbar^2}{2}\sum_{k+l=m-1}\frac{(2k+1)(2l+1)}{2^{m+1}}\frac{\pd^2}{\pd p_{2k+1}\pd p_{2l+1}}.
\end{align*}
For this case, the equivalence of the topological recursion and the Virasoro constraints was proved by Zhou~\cite{Zhou13}, we refer the reader to that paper for more details.

\subsection{Deformed $r$-Bessel curve}
\label{sec:geo-negative-rspin}
We consider the following $1$-parameter family of spectral curves, called the $\epsilon$-deformed $r$-Bessel curve:
$$
\cC^{r,\epsilon}=\big(\, \mathbb P^1,\quad x(z)=-z^r+\epsilon z, \quad y(z)=\tfrac{\sqrt{-r}}{z}\, \big).
$$	
According to~\cite{CGG22} (see also~\cite{GJZ23}), up to some constant terms, the CohFT associated with $\cC^{r,\epsilon}$ coincides with the CohFT of the deformed negative $r$-spin theory, introduced by Norbury~\cite{Nor17} for $r=2$ without deformation ($\epsilon=0$) and generalized by Chidambaram, Garcia-Failde and Giacchetto in~\cite{CGG22} for arbitrary integer $r\geq 2$ with deformation parameter $\epsilon$.

The geometric descendent generating series $\cD^{r,\epsilon}({\bf t};\hbar)$ for the deformed negative $r$-spin theory requires the knowledge of the $S$-matrix and $\dvac$-vector (equivalently, the $J$-function), which beyond the information of the spectral curve and the topological recursion.
The formula of $\cD^{r,\epsilon}({\bf t};\hbar)$ is derived in~\cite[Theorem 6.8]{GJZ23} by determining the $S$-matrix and $J$-function via the twisted theory (see~\cite[Appendix B.2]{GJZ23}). 
By using these explicit computations, in~\cite[Proposition 4.2]{GZ25}, the authors proved the following geometric descendent Virasoro constraints for the deformed negative $r$-spin theory:
	$$
	L^{r,\epsilon}_{m}\cD^{r,\epsilon}(\mathbf t;\hbar)=0, \qquad  m\geq 0.
	$$
Here the operators $L^{r,\epsilon}_{m}$, $m\geq 0$, have the following explicit formulae:
\begin{align*}
L^{r,\epsilon}_m=&\, \delta_{m,0}\frac{(r-1)\epsilon^2}{2\, \hbar^2}+\delta_{m,0}\frac{r^2-1}{24\, r}
+\sum_{k\geq 0}\sum_{a=1}^{r-1}\frac{\Gamma(m+1+k+\frac{a}{r})}{\Gamma(k+\frac{a}{r})}\tilde t_k^{a}\frac{\pd }{\pd t_{k+m}^{a}}\\
&\, +\frac{\hbar^2}{2}\sum_{k=1}^{m}(-1)^k\sum_{a=1}^{r-1}
\frac{\Gamma(m+1-k+\frac{a}{r})}{\Gamma(-k+\frac{a}{r})}
\frac{\pd^2 }{\pd t_{k-1}^{r-a}\pd t_{m-k}^{a}},
\end{align*}
where $\tilde t_k^a=t_{k}^{a}-\delta_{k,0}\delta_{a,r-1}\cdot r$.

For the TR side, we note firstly that there is only one boundary point $z=\infty$.
We choose the local coordinate $\lambda$ near the boundary point by 
$$
x(z)=\lambda^{r},
$$
that satisfies $\lim_{z=\infty} \lambda/ z=\xi=e^{-\pi i/r}$, and we denote the corresponding generating series of TR descendents by $Z^{r,\epsilon}({\bf p};\hbar)$.
By using $\lambda^{-1}=\xi^{-1}\cdot z^{-1}\cdot  (1-\epsilon\cdot z^{-(r-1)})^{-1/r}$, it is easy to see
$$
ydx=r\sqrt{-r}\xi \lambda^{r-2}d\lambda-\sqrt{-r}\epsilon \lambda^{-1}d\lambda+dO(\lambda^{-(r-1)}).
$$
Hence $v_{r-1}=r\sqrt{-r}\xi$, and $v_0=-\sqrt{-r}\epsilon$.
Furthermore, $c_m=\delta_{m,0}c_0$, where
$$
c_0=r\mathop{\Res}_{z=0}\, (-z^r+\epsilon z)\frac{1}{z^2} (-rz^{r-1}+\epsilon)dz
=r\epsilon^2.
$$
Notice that $y(z)$ has pole at $z=0$, and for $m\geq 0$, $x(z)^{m+1}y(z)$ has single pole at boundary point $z=\infty$,
by Theorem~\ref{thm:TR-virasoro-descendent-intro}, the generating series $Z^{r,\epsilon}({\bf p};\hbar)$ satisfies the Virasoro constraints
$$
\cL^{r,\epsilon}_m Z^{r,\epsilon}({\bf p};\hbar)=0, \qquad m\geq 0.
$$
Here the operators $\cL^{r,\epsilon}_{m}$, $m\geq 0$, are given by the following explicit formulae:
\begin{align*}
\cL^{r,\epsilon}_m=&\, \delta_{m,0}\frac{(r-1)\epsilon^2}{2\, \hbar^2}+\delta_{m,0}\frac{r^2-1}{24\, r}
+\frac{1}{r}\sum_{k\geq 0}\sum_{a=1}^{r-1}\tilde p_{rk+a}(rm+rk+a)\frac{\pd }{\pd p_{rm+rk+a}}\\
&\, +\frac{\hbar^2}{2r}\sum_{k=1}^{m}\sum_{a=1}^{r-1} (rm-rk+a)(rk-a)
\frac{\pd^2 }{\pd p_{rk-a}\pd p_{rm-rk+a}},
\end{align*}
where $\tilde p_k=p_k-\delta_{k,r-1}\cdot r\sqrt{-r}\xi$.

Due to the result of~\cite{GJZ23}, we have~\footnote{We note that the definition of $Z^{r,\epsilon}({\bf p};\hbar)$ in the present paper differs slightly from that in~\cite{GJZ23}: here, the definition includes the invariants $\<-\>_{0,1}$, whereas these invariants were excluded in~\cite{GJZ23} since they do not affect the $r$KdV integrability.
In addition, the variables $p_k$ used here differ from those in~\cite{GJZ23} by a factor of $\xi^k$, that is, $p_k^{\rm GJZ}=\xi^k p_k$.}
$$
e^{\sum_{g\geq 2}\frac{\hbar^{2g-2}}{(\sqrt{-r}\epsilon)^{2g-2}}\frac{B_{2g}}{2g(2g-2)}}\cdot\cD^{r,\epsilon}({\bf t(p)};\hbar)
= Z^{r,\epsilon}({\bf p};\hbar),
$$
where $B_{2g}$ is the Bernoulli number and the coordinate transformation is given by
$$
t_n^i({\bf p})=-\frac{\xi^i}{\sqrt{-r}}\frac{\Gamma(\frac{i}{r}+n)}{\Gamma(\frac{i}{r})}p_{i+rn}.
$$
We see for the deformed $r$-Bessel curve, the geometric descendent Virasoro constraints and the TR descendent Virasoro constraints are equivalent.

\subsection{Spectral curve of  extended Grothendieck's dessins d'enfants}
Consider spectral curve
$$\textstyle
\cC^{\rm eGdd}=(\mathbb P^1,\quad x(z)=z+\frac{uv}{z}+u+v,\quad y(z)=\frac{1}{2}-\frac{u}{2}\frac{1}{z+u}-\frac{v}{2}\frac{1}{z+v}).
$$
It is straightforward to compute the $R$-matrix (by equation~\eqref{def:EO-R}) and the vacuum vector (by equation~\eqref{eqn:TR-vacuum}) of the CohFT associated with $\cC^{\rm eGdd}$, and to prove this CohFT is exactly the CohFT $\Omega^{\rm eGdd}$ for the extended Grothendieck's dessins d'enfants that introduced in~\cite[\S 5]{GZ25} (the $R$-matrix and the vacuum vector here correspond to those in~\cite{GZ25}, evaluated at $\tau=0$, and $u,v=\epsilon_{1,2}$ in~\cite{GZ25}).
Therefore, we call $\cC^{\rm eGdd}$ the spectral curve of  extended Grothendieck's dessins d'enfants.

We use the $S$-matrix and $\dvac$-vector introduced in~\cite{GZ25} (evaluated at $\tau=0$) to define the formal geometric total descendent potential $\cD^{\rm eGdd}({\bf t};\hbar)$, then we have the geometric descendent Virasoro constraints~\cite[Proposition 5.1]{GZ25}:
$$
	L^{\rm eGdd}_m\cD^{\rm eGdd}({\bf t};\hbar)=0,\qquad m\geq 0.
	$$
Here the operators $L^{\rm eGdd}_{m}$, $m\geq 0$, have the following explicit formulae:
\begin{align*}
L^{\rm eGdd}_m=&\, \delta_{m,0}\cdot\frac{(t_0^1+u)(t_0^1+v)}{\hbar^2} 
+\sum_{k\geq 0}\frac{(k+m+1)!}{k!}\tilde t_{k}^{0}\frac{\pd }{\pd t_{k+m}^0}
+\hbar^2\sum_{k+l=m}k!\, l!\, \frac{\pd^2}{\pd t_{k-1}^0\pd t_{l-1}^0}\\
&\, +2m!\tilde t_0^1\frac{\pd}{\pd t_{m-1}^0}
+2\sum_{k\geq 1}\sum_{i=0}^{m}\frac{(k+m)!}{(k-1)!(k+i)}t_k^1\frac{\pd}{\pd t_{k+m-1}^0}
+\sum_{k\geq 1}\frac{(k+m)!}{(k-1)!} t_{k}^{1}\frac{\pd }{\pd t_{k+m}^1},
\end{align*}
where $\tilde t_k^a=t_k^a-\delta_{k,0}\delta_{a,0}+\delta_{k,0}\delta_{a,1}\frac{u+v}{2}$.

For the TR side, there are two boundary points for $\cC^{\rm eGdd}$: $z=0$ and $z=\infty$.
We define the corresponding local coordinates, $\lambda_0$ and $\lambda_{\infty}$, by
$$
x(z)|_{0}=\lambda_0,\qquad
x(z)|_{\infty}=\lambda_{\infty},
$$
and denote the corresponding generating series of TR descendents by $Z^{\rm eGdd}({\bf p};\hbar)$.
By direct computations, one can see
\begin{align*}
ydx|_{0}=&\, -\frac{1}{2}d\lambda_0+\frac{u+v}{2} \lambda_0^{-1}d\lambda_0 + dO(\lambda^{-1}_0),\\
ydx|_{\infty}=&\, \frac{1}{2}d\lambda_{\infty}-\frac{u+v}{2}\lambda_{\infty}^{-1}d\lambda_{\infty}+dO(\lambda^{-1}_{\infty}).
\end{align*}
Hence $v^{0}_{1}=-\frac{1}{2}$, $v^{0}_{0}=\frac{u+v}{2}$ and $v^{\infty}_{1}=\frac{1}{2}$, $v^{\infty}_{0}=-\frac{u+v}{2}$.
Furthermore, $c_m=\delta_{m,0}c_0$, where
$$
c_0=-\Big(\mathop{\Res}_{z=-u}+\mathop{\Res}_{z=-v}\Big)\, x(z) y(z)^2dx(z)=-\frac{(u-v)^2}{2}.
$$
Notice that $y(z)$ has poles at $z=-u$ and $z=-v$, and for $m\geq 0$, $x(z)^{m+1}y(z)$ has single poles at boundary points $z=0$ and $z=\infty$,
by Theorem~\ref{thm:TR-virasoro-descendent-intro}, the generating series $Z^{\rm eGdd}({\bf p};\hbar)$ satisfies the Virasoro constraints:
$$
\cL_m^{\rm eGdd}Z^{\rm eGdd}({\bf p};\hbar)=0,\qquad  m\geq 0.
$$
Here the operators $\cL^{\rm eGdd}_{m}$, $m\geq 0$, have the following explicit formulae:
\begin{align*}
\cL_m^{\rm eGdd}=&\, \delta_{m,0}\frac{uv}{\hbar^2} +\frac{m(u+v)}{2}\Big(\frac{\pd}{\pd p_{m}^{\infty}}-\frac{\pd}{\pd p_{m}^{0}}\Big)
+\sum_{i=0,\infty}\sum _{k\geq 1} \tilde p^i_k (m+k)\frac{\pd}{\pd p^i_{m+k}} \\
&\, +\frac{\hbar^2}{2}\sum_{i=0, \infty}\sum_{k+l=m}kl\frac{\pd^2}{\pd p^i_k\pd p^i_l},
\end{align*}
where $\tilde p_k^{0}=p_k^{0}+\delta_{k,1}\frac{1}{2}$, $\tilde p_k^{\infty}=p_k^{\infty}-\delta_{k,1}\frac{1}{2}$.

By applying a method similar to that used for the deformed $r$-Bessel curve in~\cite{GJZ23}, one can show that $\cD^{\rm eGdd}({\bf t(p)};\hbar)=Z^{\rm eGdd}({\bf p};\hbar)$, where the coordinate transformation ${\bf t(p)}$ is given by
$$
t_k^0=k!(p_{k+1}^{\infty}-p^0_{k+1}),\quad 
t_k^1=0,\qquad k\geq 0.
$$
It is easy to see that via this coordinate transformation, the operators $\cL^{\rm eGdd}_m$ coincide with the operators $L_{m}^{\rm eGdd}$ for $m\geq 0$.
In other words, the TR descendent Virasoro constraints coincide with the geometric ones.

\begin{remark}
It is clear that, for this case, the generating series $Z^{\rm eGdd}({\bf p};\hbar)$ does not contains the full geometric descendent information of the $S$- and $\dvac$-calibrated CohFT $\Omega^{\rm eGdd}$.
In fact, $Z^{\rm eGdd}({\bf p};\hbar)=Z^{\rm Gdd}({\bf p'(p)};\hbar)$, where $Z^{\rm Gdd}({\bf p'};\hbar)$ is the generating series of dessin counting, and the coordinate transformation is given by $p'_{k+1}=p_{k+1}^{\infty}-p^0_{k+1}$.
The equivalence of this topological recursion and the Virasoro constraints for $Z^{\rm Gdd}$ is first proved by Kazarian--Zograf~\cite{KZ15}.
\end{remark}

\subsection{Weierstrass curve}
\label{subsec:weierstrass}
We consider the following Weierstrass curve: 
$$
\cC=\big(\Sigma=\mathbb C/(\mathbb Z+\tau\mathbb Z),\quad
x(z)=\wp(z;\elltau),\quad
y(z)=\wp'(z;\elltau)
\big),
$$
where $\wp$ is the Weierstrass P-function.
It is well known that $x,y$ satisfy 
$$
y^2=4x^3-g_2(\tau)x-g_3(\tau).
$$
where $g_2=60\, G_4$, $g_3=140\, G_6$, and $G_{2k}$, $k\geq 1$, is the weight-$2k$ holomorphic Eisenstein series.
We have critical points $z^1=\frac{1}{2}$, $z^2=\frac{\tau}{2}$, and $z^3=\frac{1+\tau}{2}$, the corresponding critical values are
$$
x^\beta=u^\beta(\tau), \qquad \beta=1,2, 3,
$$
where $u^\beta(\tau)=\wp(z^\beta;\tau)$.
For the canonical basis $e_\beta$, we have
$$
\Delta_{\beta}^{-1}=\eta(e_\beta,e_\beta)=\tfrac{ y'(z^\beta)^2}{ x''(z^\beta)}
=6 (u^\beta)^2-\tfrac{g_2}{2}.
$$

For the geometric side, we introduce basis $\{\phi_i\}_{i=0}^{2}$, called the flat basis, as follows:
\begin{align*}
\phi_0=&\,\textstyle -2\pi {\bf i}  \sum_\beta \Delta_\beta e_\beta, \\
\phi_1=&\,\textstyle -\sqrt{2} \sum_{\beta}(u^\beta+G_2)\Delta_\beta e_\beta,\\
\phi_2=&\,\textstyle -\frac{1}{\pi {\bf i}}\sum_\beta \big((u^\beta+G_2)(u^\beta-2G_2)-3\pi i \pd_{\tau}G_2\big)\Delta_\beta e_\beta.
\end{align*}
By using relations:
$$ 
u^1+u^2+u^3=0,\quad
u^1u^2+u^2u^3+u^1u^3=-g_2/4,\quad
u^1u^2u^3=g_3/4,
$$
and by direct computations, we have
$$
\eta(\phi_i,\phi_j)=\delta_{i+j=2}.
$$
Under the canonical basis $\{e_\beta\}$, $\E=\diag\{x^1,x^2,x^3\}$, we have under the flat basis $\{\phi_i\}$,
$$
\E=\left(\begin{array}{ccc}
	- G_2 & \frac{3 }{\sqrt{2}} \pd_\tau G_2 &  \pd_{\tau}^2G_2 \\
	\sqrt{2} \pi {\bf i} &  2\, G_2  & \frac{3 }{\sqrt{2}}\pd_{\tau} G_2 \\
	0  & \sqrt{2} \pi {\bf i} & - G_2 
\end{array}\right).
$$
Here we have used the Ramanujan identities of Eisenstein series,
$$\textstyle
\pd_{\elltau}G_2=\frac{5G_4-G_2^2}{2\pi {\bf i}},\qquad
\pd_{\elltau}G_4=\frac{7G_6-2G_2G_4}{\pi{\bf i}},\qquad
\pd_{\elltau}G_6=\frac{30G_4^2-21G_2G_6}{7\pi{\bf i}}.
$$
We define operator $\mu$ by $\mu(\phi_i)=(\tfrac{i}{2}-\tfrac{1}{2})\phi_i$.
\begin{proposition}
The CohFT $\Omega$ associated with the Weierstrass curve is homogeneous with conformal dimension $\delta=-2$ and is reconstructed by
$$
\Omega=R\cdot T\cdot (\oplus_{\beta=1}^{3} \Omega^{\rm KW}_{\beta}),
$$
where the $R$-matrix is determined by the homogeneity condition
$$
[R_{m+1},\E]=(m+\mu)R_m,
$$
and the $T$-vector is determined by
$$
T(z)=z(\bar {\bf 1}-R^{-1}(z)\vac(z)).
$$
Here  $\vac(z)={\bf 1}+\vac_1z$ is the vacuum vector, with the non-flat unit
$$
{\bf 1}=e_1+e_2+e_3=-3\,  \pd_\tau G_2\, \phi_0 -3\sqrt{2}\, G_2\, \phi_1-6\pi {\bf i} \,  \phi_2,
$$
and $\vac_1=-\tfrac{3\sqrt{2}}{2}\phi_1$.
\end{proposition}
\begin{proof}
Since the proof of this Proposition is similar as the discussion of the spectral curve for $M_{1,1}$ in~\cite[\S 7.1]{GJZ23}, we omit the details of the explicit computations and just show the sketch of the proof.
First, one computes $R_1$ using equation~\eqref{eqn:R1}, and obtains its representation in the flat basis.
It can then be verified that the operator $\mu$ satisfies $\mu=[R_1,\E]$.
Second,  by checking equation~\eqref{eqn:Dzeta-Psi-d}, one sees that the spectral curve is homogeneous.
Then, by Lemma~\ref{lem:hom-TR}, all $R_k$ for $k\geq 1$ are determined.
Thirdly, the vacuum vector can be computed directly from equation~\eqref{eqn:TR-vacuum}, leading to
$$\textstyle
\E\vac_1+(\mu+\frac{\delta}{2})\vac_0=0,\qquad
(\mu+\frac{\delta}{2}+1)\vac_1=0,
$$
where $\delta=-2$. The homogeneity of the CohFT follows from the homogeneity of $R$-matrix and vacuum vector. The Proposition is proved.
\end{proof}

We choose $S$-calibration by requiring the homogeneity condition
$$
\E S_{k-1}=kS_k +[S_k,\mu]+S_{k-1}\rho,
$$
and initial condition $(S_1)^2_0=\oint_{B}\oint_B\omega_{0,2}=2\pi {\bf i}\,  \tau$, where $\rho=0$.
We choose $\dvac$-calibration by $\dvac(\givz)=\vac(\givz)$, then we have the $J$-function
$$
J(-\givz)=-\givz S^{*}(-\givz)\dvac(\givz)=\tfrac{3\sqrt{2}}{2}\, \phi_1 \givz^2-\tfrac{1}{10}\pd_{\tau}g_2\, \phi_0-\tfrac{1}{2\sqrt{2}}g_2\, \phi_1+O(\givz^{-1}).
$$
We define the total descendent potential by the generalized Kontsevich--Manin formula~\cite{GZ25} (see~\cite{KM98} and ~\cite{Giv01a} for the original Kontsevich--Manin formula):
$$
\cD({\bf t};\hbar)=e^{\frac{1}{\hbar^2}J_{-}({\bf t})+\frac{1}{\hbar^2}W({\bf t,t})}\cA([S(z){\bf t}(z)]_{+};\hbar).
$$
Then it follows immediately from~\cite[Theorem 2]{GZ25} that the geometric total descendent potential $\cD({\bf t};\hbar)$ satisfies the following Virasoro constraints
$$
L_m\cD({\bf t};\hbar)=0,\qquad m\geq -1.
$$
Here the operators $L_{m}$, $m\geq -1$, are given by 
\begin{align*}
	L_m=&\, \frac{\delta_{m,-1}}{2\hbar^2}\, \eta(\tilde t_0, \tilde t_0)
	+\frac{\delta_{m,0}}{16}
	+\sum_{k \geq 0}\sum_{a=0}^{2} \frac{\Gamma(\frac{a}{2}+k+m+1)}{\Gamma(\frac{a}{2}+k)} \tilde t_k^a \frac{\pd }{\pd t_{k+m}^{a}}    \nonumber \\
	&\,+\frac{\hbar^2}{2}\sum_{k=0}^{m-1}\frac{(2k+1)!!(2m-2k-1)!!}{2^{m+1}}\frac{\pd^2 }{\pd t_{k}^{1}\pd t_{m-k-1}^{1}},
\end{align*}
	where $\tilde t_k^a=t_k^a$ except that $\tilde t_0^0=t_0^0-\tfrac{1}{10}\pd_{\tau}g_2$, $\tilde t_0^1=t_0^1-\tfrac{1}{2\sqrt{2}}g_2$, and $\tilde t_2^1=t_2^1+\tfrac{3\sqrt{2}}{2}$.

\begin{remark}
The $S$-calibration is well-defined according to the discussions in~\cite[\S 2]{GZ25}.
We have precisely, $(k+\mu_b-\mu_a)(S_k)^a_b=(\E S_{k-1}-S_{k-1}\rho)^a_b$, $k\geq 1$.
For $k=1$, 
$$
S_1=\left(\begin{array}{ccc}
-G_2 & \sqrt{2} \, \pd_\tau G_2 & \frac{1}{2} \, \pd_{\tau}^2 G_2 \\
2\sqrt{2}\pi {\bf i}  & 2\, G_2 & \sqrt{2} \, \pd_{\tau} G_2 \\
2\pi {\bf i}\, \tau & 2\sqrt{2}\pi {\bf i}  & -G_2 
\end{array}\right).
$$
\end{remark}

We turn to the TR side. Clearly, there is only one boundary point $z=0$.
Recall that by the definition,
$$
d\zeta^{\bar \beta}(z)=-\mathop{\Res}_{z=z^\beta}\frac{B(z',z)}{\sqrt{x''(z^\beta)}(z'-z^\beta)}
=-\frac{1}{\sqrt{x''(z^\beta)}}\big(\wp(z-z^\beta;\tau)+G_2\big)dz,
$$
where $x''(z^\beta)=6(u^\beta)^2-\tfrac{g_2}{2}$.
By taking integration, we get
$$
\zeta^{\bar\beta}(z)=\int_0^{z}d\zeta^{\bar\beta}(z)=\frac{1}{\sqrt{x''(z^\beta)}}\big(\zeta(z-z^\beta;\tau)+\zeta(z^\beta;\tau)-G_2z\big),
$$
where $\zeta(z;\tau)$ (no super-index) is the Weierstrass Zeta function.
By using the additional formula of $\zeta(z;\tau)$,
we have
$$
\zeta^{\bar\beta}(z)=\frac{1}{\sqrt{x''(z^\beta)}}\bigg(\zeta(z;\tau)-G_2z+\frac{1}{2}\frac{\wp'(z;\tau)}{\wp(z;\tau)-u^\beta}\bigg).
$$
Define $\zeta^{i}(z)$ by $\sum_i \phi_i\zeta^i(z)=\sum_\beta \bar e_{\beta}\zeta^{\bar \beta}$, it is straightforward to compute that
\begin{align*}
\zeta^0(z)=&\,\textstyle  \frac{1}{2\pi {\bf i}}\big(G_2z-\zeta(z)+\frac{1}{\wp'(z)} (-2\wp(z)^2+2G_2\wp(z)+G_2^2+15G_4)\big),\\
\zeta^1(z)=&\,\textstyle  -\sqrt{2}\, \frac{\wp(z)+G_2}{\wp'(z)} , \\
\zeta^2(z)=&\,\textstyle  -\frac{2\pi {\bf i}}{\wp'(z)}.
\end{align*}
Then the functions $\zeta_k^i=(-\frac{d}{dx(z)})^k\zeta^i(z)$ are clear.
Near the boundary point $z=0$, for each $i=0,1,2$, we introduce the local functions
$$
\chi^i(z)=\sum_{k\geq 0}(-1)^k (S^{*}_k)^i_j\zeta^j_{k}(z).
$$
Similarly to the discussion for the spectral curve of the Hurwitz space $M_{1,1}$ in~\cite[\S 7]{GJZ23}, by using the homogeneity condition of $S$-matrix and equation~\eqref{eqn:Dzeta-Psi-d} (it is straightforward to check that this equation holds for the Weierstrass curve),
we have $\chi^0=\chi^2=0$, and
$$\textstyle
\chi^1=\frac{1}{\sqrt{2}}\lambda^{-1},
$$
where $\lambda$ is the local coordinate near the boundary point, defined by
$$
x(z)=\lambda^2,\qquad 
\lim_{z\to 0}\tfrac{\lambda^{-1}}{z}=1.
$$
This gives
$$\textstyle
\chi_k^1=\big(-\tfrac{d}{2\lambda d\lambda}\big)^k\chi^1 =\frac{(2k-1)!!}{2^{k+\frac{1}{2}}}\lambda^{-2k-1}.
$$
Furthermore, for this choice of the local coordinate $\lambda$, we denote the corresponding TR descendent generating series by $Z({\bf p};\hbar)$, then above discussions shows that
$$
\cD({\bf t(p)};\hbar)=Z({\bf p};\hbar),
$$
where the coordinate transformation is given by
\beq\label{eqn:weierstrass-t=tp}
t_k^0=t_k^2=0,\quad
t_k^1=\frac{(2k-1)!!}{2^{k+\frac{1}{2}}} \, p_{2k+1},\qquad k\geq 0.
\eeq
Here we have used~\cite[Proposition 3.7]{GJZ23} and the validity of~\cite[Conjecture 0.4]{GJZ23} (which can be proved by following a similar argument as the one used for the spectral curve of the Hurwitz space $M_{1,1}$).

Now we consider the expansion of $ydx$ near the boundary point, by using the local coordinate $\lambda$, we have
$$
y(z)dx(z)|_{z=0}=-4\lambda^4d\lambda+\frac{g_2}{2}d\lambda+d(O(\lambda^{-1})).
$$
We see the non-vanished terms of $v_k$ are $v_5=-4$ and $v_1=\frac{g_2}{2}$.
By Theorem~\ref{thm:TR-virasoro-descendent-intro}, we have the following TR descendent Virasoro constraints:
$$
\cL_mZ({\bf p};\hbar)=0,\qquad m\geq -1.
$$
Here the operators $\cL_{m}$, $m\geq -1$, are given by 
\begin{align*}
\cL_m=&\, \delta_{m,-1}\frac{(\tilde p_1)^2}{4\hbar^2}+\delta_{m,0}\frac{1}{16}
+\frac{1}{2}\sum_{k\geq 0}\tilde p_{2k+1}(2k+2m+1)\frac{\pd }{\pd p_{2k+2m+1}}\\ 
&\, +\frac{\hbar^2}{4}\sum_{k+l=m-1}(2k+1)(2l+1)\frac{\pd^2}{\pd p_{2k+1}\pd p_{2l+1}}.
\end{align*}
where $\tilde p_k=p_k+\delta_{k,1}\frac{g_2}{2}-4\, \delta_{k,5}$.
It is easy to see that via the coordinate transformation~\eqref{eqn:weierstrass-t=tp}, the operators $\cL_m$ coincide with the operators $L_{m}$ for $m\geq -1$.
In other words, the TR descendent Virasoro constraints coincide with the geometric ones.

For the non-perturbative topological recursion, we have the vector field $\varphi$, defined by~\eqref{def:filling-fraction}, is equal to $\phi_0$.
By~\cite[Proposition 4.8]{GJZ23} (together with~\cite[Definition 4.7]{GJZ23}, and also note that the definition of $Z^{\rm NP}_{\mu,\nu}$ is slightly modified in the present paper), the non-perturbative generating series $Z^{\rm NP}_{\mu,\nu}({\bf p};\hbar;w)$ has the following formula:
$$
Z^{\rm NP}_{\mu,\nu}({\bf p};\hbar;w)=\cD({\bf t}+\tfrac{\hbar}{2\pi{\bf i}}\phi_0\cdot\pd_w;\hbar)\, \theta[\cmn](w-\tfrac{\hbar}{2\pi{\bf i}}\<\phi_0\>_{0,1};\tau')|_{{\bf t=t(p)},\, \tau'=0}.
$$
where $\<\phi_0\>_{0,1}=2\pi {\bf i}\tau\cdot (-\frac{\pd_\tau g_2}{10})+\frac{16\pi {\bf i}}{5\sqrt{2}}\cdot (-\frac{g_2}{2\sqrt{2}})$.
We observe that the non-perturbative contributions modify the geometric descendent Virasoro constraints, but do not affect the TR descendent Virasoro constraints.
We note here that, by~\cite[Theorem I]{GJZ23}, $Z^{\rm NP}_{\mu,\nu}({\bf p};\hbar;w)$ is a tau-function of the KdV hierarchy.
We also refer the reader to~\cite{Iwa20} for the relation between this example and the first Painlev\'e equation.

\addtocontents{toc}{\protect\setcounter{tocdepth}{0}}

\end{document}